\documentclass[aps,pra,superscriptaddress,reprint]{revtex4-2}

\usepackage{amsmath,amssymb,amsfonts,bm}
\usepackage{mathtools}
\usepackage{physics}
\usepackage{graphicx}
\usepackage{xcolor}
\usepackage[colorlinks=true,linkcolor=black,citecolor=black,urlcolor=black]{hyperref}
\mathtoolsset{showonlyrefs=true}
\usepackage{comment}

\usepackage{amsthm}
\theoremstyle{plain}
\newtheorem{theorem}{Theorem}
\newtheorem{proposition}[theorem]{Proposition}
\newtheorem{lemma}[theorem]{Lemma}
\newtheorem{corollary}[theorem]{Corollary}
\theoremstyle{definition}
\newtheorem{definition}[theorem]{Definition}

\theoremstyle{remark}

\makeatletter
\def\frontmatter@RRAP@format{%
  \small
  \centering
}%
\def\punct@RRAP{; }%
\makeatother

\newcommand{\X}{\mathcal{X}}

\newcommand{\Dcal}{\mathcal{D}}
\newcommand{\Kcal}{\mathcal{K}}
\newcommand{\Bcal}{\mathcal{B}}
\newcommand{\Hcal}{\mathcal{H}}
\newcommand{\E}{\mathbb{E}}
\newcommand{\Id}{I}
\newcommand{\OPT}{\mathrm{OPT}}
\newcommand{\Cut}{\mathrm{Cut}}

\begin{document}

\title{Decoder-Consistent Hamiltonians for Quantum Relaxations via POVM Pullbacks}
\author{Takayuki Suzuki}
\affiliation{Intellectual Property \& Technology Strategy Division, SCSK Corporation, Koto, Tokyo 135-8110, Japan}
\date[]{\today}

\begin{abstract}
In compression-based quantum relaxations such as quantum random access optimization (QRAO), a quantum state is optimized on a reduced number of qubits and then decoded into a classical solution on which the objective function is evaluated. The Hermitian operator whose expectation value is optimized should therefore be consistent with the expected objective value after decoding. For a fixed single-shot POVM decoder, we define the decoder-consistent Hamiltonian as the pullback of the classical objective observable through the decoder. Its expectation value exactly equals the expected post-decoding objective value for every quantum state. We further show that any Hermitian operator inducing the same ordering over the full quantum state space must be a positive affine transformation of the decoder-consistent Hamiltonian. Using a Boolean Fourier expansion, the decoder-consistent Hamiltonian decomposes into pullbacks of individual Fourier components. This provides a common framework for comparing direct encoding, simultaneous-decoding POVMs for QRAC blocks, and the effective POVM induced by the magic rounding of Teramoto et al. in terms of which components they preserve and at what rates. For standard magic-state rounding in coloring-based QRAO for MaxCut, the framework recovers the known positive affine relation between the relaxed-Hamiltonian expectation and the expected rounded cut value. For the simultaneous-decoding POVMs analyzed here, we further show that the analogous relation generally fails for mixed-degree QUBO objectives because one- and two-body components are attenuated at different rates. We also present a systematic construction of POVM decoder blocks from prescribed Fourier components and derive corresponding lower bounds on the expected post-decoding objective value under explicit conditions.
\end{abstract}

\maketitle

\section{Introduction}
\label{sec:introduction}

Combinatorial optimization is among the leading candidate application areas for near-term quantum computing. The quantum approximate optimization algorithm (QAOA) and related variational algorithms encode a classical objective function into a cost Hamiltonian and optimize its expectation value within a parameterized family of quantum states~\cite{farhi2014,hadfield2019,zhu2022}.
However, due to limited qubit counts, high measurement costs, and noise in currently available quantum devices, the problem sizes accessible via one-to-one direct encoding are limited.

For this reason, quantum relaxation methods have been studied, where multiple classical variables are compressed into a smaller number of qubits, and after variationally optimizing the expectation value of a quantum Hamiltonian, the resulting quantum state is decoded into a classical solution. For example, quantum random access optimization (QRAO), based on quantum random access codes (QRACs)~\cite{ambainis1999,nayak1999}, was proposed in Ref.~\cite{fuller2024}. In QRAO-type methods, multiple classical variables are assigned to noncommuting Pauli operators, allowing a Hamiltonian to be constructed on fewer qubits than in direct encoding. Related extensions and comparative studies have since been actively explored~\cite{fuller2024,teramoto2023,he2025,kondo2025}.

In this class of compression-based methods, the decoder that converts quantum states into classical solutions and the Hamiltonian used in variational quantum algorithms are often specified as separate components of the optimization procedure. This procedural separation does not by itself imply a mismatch. In particular, for MaxCut with magic-state rounding, Fuller et al. established that the expected rounded cut is linearly related to the relaxed-Hamiltonian expectation, and a related affine relation was subsequently made explicit in the RQRAO framework~\cite{fuller2024,kondo2025}. The general question is therefore how to characterize, for a fixed decoder and a general classical objective, the Hermitian observable whose expectation equals the expected decoded objective.

We provide a general framework for characterizing the relation between a fixed decoder and the quantum objective used in compression-based quantum optimization. The framework identifies the decoder-consistent observable, determines when another Hermitian objective induces the same ordering of quantum states, and turns Fourier-component pullbacks into a principle for decoder design. Our focus is this structural characterization rather than a speedup for a particular problem class.

We consider fixed single-shot decoders, including fixed classical post-processing of measurement outcomes, and represent each such decoder as a POVM. Using the decoder's dual map, we pull back the classical objective function to a Hermitian operator on the quantum Hilbert space. For a POVM decoder $\Dcal$ with effects $\{M_s\}_{s\in\X}$ and an objective function $f:\X\to\mathbb R$, we write
\[
H_{\Dcal}[f]:=\sum_{s\in\X}f(s)M_s,
\]
and call it the decoder-consistent Hamiltonian. Its expectation value exactly equals the expected objective value after decoding for every quantum state. Importantly, the quantity optimized here is the expected post-decoding objective value, not the probability that the decoder outputs an optimal solution.

Furthermore, we show that this Hamiltonian is essentially unique unless it is proportional to the identity: any Hermitian operator that induces the same ordering over all quantum states is a positive affine transformation of $H_{\Dcal}[f]$. Thus, once the decoder and objective function are fixed, the Hermitian operator whose expectation value is to be optimized is uniquely determined up to a positive affine transformation.

%This requirement differs from the code-state reproduction condition often used in conventional QRAO-type relaxations. Conventional constructions rescale the Hamiltonian so that its expectation value on each code state reproduces the corresponding classical objective value, whereas our condition requires the Hamiltonian expectation value to equal the expected objective value after a fixed decoder for every quantum state. From this decoder-first viewpoint, we can directly compare the expectation value of a conventional QRAO-type relaxation Hamiltonian with the expected objective value actually obtained after decoding the quantum state into a classical solution.

This decoder-first requirement differs from the code-state reproduction condition used in conventional QRAO-type relaxations. Whereas conventional constructions reproduce the classical objective on code states, decoder consistency reproduces the expected post-decoding objective value for every quantum state. The resulting criterion determines whether a code-state-compensated relaxation and a specified decoder induce the same ordering over the full quantum state space; Section~\ref{sec:qrao-reinterpretation} applies it to selected QRAO constructions.

An explicit matrix representation of $H_{\Dcal}[f]$ is not required in order to optimize its expectation value. If, at each variational parameter setting, one samples classical outputs from the fixed decoder and uses the sample mean of their objective values, then this sample mean is an unbiased estimator of the decoder-consistent expectation. Explicitly, for independent outputs $S_1,\ldots,S_N\sim\Dcal(\rho)$,
\[
\widehat F_N(\rho):=\frac1N\sum_{r=1}^N f(S_r),
\qquad
\mathbb E[\widehat F_N(\rho)]=\Tr(H_{\Dcal}[f]\rho).
\]
Thus, decoder consistency may be implemented implicitly through sampling. The comparison developed below is relevant when a separately specified relaxation Hamiltonian is optimized and decoding is applied only as a subsequent rounding or evaluation step.

The main contributions of this paper are fourfold. First, for an arbitrary fixed single-shot POVM decoder and objective function, we define a decoder-consistent Hamiltonian whose expectation value exactly represents the expected objective value after decoding, and establish its order-preserving uniqueness over the full quantum state space. Second, by decomposing the classical objective into Fourier components, we compare representative constructions, including direct encoding, simultaneous-decoding POVMs for QRAC blocks~\cite{fuller2024}, and the POVM induced by the magic rounding of Teramoto et al.~\cite{teramoto2023}, in terms of which components they preserve and at what rates. Third, we recover the known affine consistency of standard MaxCut magic-state rounding and identify the degree-dependent obstruction that generally prevents the same consistency for mixed-degree QUBO objectives under the simultaneous-decoding POVMs analyzed here; the precise sector conditions are stated in Section~\ref{sec:qrao-reinterpretation}. Fourth, we present a systematic construction of POVM decoder blocks from prescribed Fourier components, together with corresponding lower bounds on the expected post-decoding objective value under specified conditions.

The remainder of this paper is organized as follows. Section~\ref{sec:framework} fixes the notation and formulates the decoder-consistent Hamiltonian and the order-uniqueness result. Section~\ref{sec:parity-qrao} introduces pullbacks of Fourier components, clarifies their relation to existing QRAO-type relaxations, and derives performance bounds from test states. Section~\ref{sec:decoder-design} presents a systematic POVM construction based on a prescribed set of components to be preserved. Section~\ref{sec:examples-implications} illustrates both decoder design and the reinterpretation of existing methods using the \((2,1)\) direct parity block and the POVM induced by the magic rounding of Teramoto et al.. Section~\ref{sec:conclusion} concludes the paper. Detailed derivations and auxiliary results are collected in the Appendices.

\section{Framework for decoder-consistent Hamiltonians}
\label{sec:framework}

In this section, we define the decoder-consistent Hamiltonian, a central concept of this paper. We also establish its expectation-value identity, uniqueness with respect to state ordering, and spectral range within a unified framework.

\subsection{Classical objective functions and POVM decoders}

\subsubsection{Classical configuration space and objective function}

We mainly use the Ising representation, with the classical configuration space \(\X=\{\pm1\}^m\), and write \([m]:=\{1,\ldots,m\}\). For a classical configuration $s=(s_1,\ldots,s_m)$, let $f:\X\to\mathbb R$ be the objective function.
When discussing QUBO objectives, we use their equivalent Ising representation obtained through the standard transformation \(x_i=(1-s_i)/2\).

Let $C$ be the classical register storing the classical configuration, with classical Hilbert space $\Hcal_C$ spanned by the orthonormal basis $\{\ket{s}_C:s\in\X\}$. We define the diagonal operator corresponding to the classical objective function by $F_C=\sum_{s\in\X}f(s)\ket{s}_C\bra{s}$.
Throughout this paper, the term classical operator refers only to an operator that is diagonal in this computational basis.

For classical spins $s_i\in\{\pm1\}$, we write the parity monomial of $A\subseteq[m]$ as $s^A:=\prod_{i\in A}s_i$ and $s^\emptyset:=1$.
Any objective function admits the unique Fourier expansion
\[
f(s)=\sum_{A\subseteq[m]}\widehat f_A s^A,
\qquad
\widehat f_A=2^{-m}\sum_{s\in\X}f(s)s^A.
\]
We refer to each term \(\widehat f_A s^A\) as a Fourier component of \(f\).
Correspondingly, defining
\[
S_{C,A}:=\sum_{s\in\X}s^A\ket{s}_C\bra{s},
\]
we have
\[
F_C=\sum_{A\subseteq[m]}\widehat f_A S_{C,A}.
\]
We also define the Fourier support of $f$ by
\begin{align}
\operatorname{supp}_{\mathrm F}(f)
:=
\left\{
A\subseteq[m]:
\widehat f_A\neq0
\right\}.
\label{eq:fourier-support}
\end{align}
Thus, $A\in\operatorname{supp}_{\mathrm F}(f)$ if and only if the Fourier component indexed by $A$ appears in $f$ with a nonzero coefficient.

We use the Fourier expansion because it separates an objective into one-body, two-body, and higher-order interaction components. Once the decoder is fixed, each component is pulled back to a corresponding quantum operator. Different components may be retained with different strengths or eliminated entirely, allowing us to characterize the action of the decoder component by component.

\subsubsection{POVM decoder}

A fixed single-shot decoder that outputs a classical configuration $s\in\X$ from a quantum state $\rho$ is represented by a POVM $\{M_s\}_{s\in\X}$, where $M_s\ge0$ and $\sum_sM_s=I$. Any fixed classical post-processing of an underlying measurement can be absorbed into these effective POVM elements.
This decoder can be expressed as a quantum-to-classical (QC) channel $\Dcal(\rho)=\sum_{s\in\X}\Tr(M_s\rho)\ket{s}_C\bra{s}$.
The dual map $\Dcal^\dagger$ is defined by $\Tr(G_C\Dcal(\rho))=\Tr(\Dcal^\dagger(G_C)\rho)$. In particular, for a classical operator $G_C=\sum_{s\in\X}g(s)\ket{s}_C\bra{s}$, we have $\Dcal^\dagger(G_C)=\sum_{s\in\X}g(s)M_s$.

In particular, $\Dcal^\dagger$ is a positive unital map from the algebra of classical operators to operators on the quantum system. Hence, for any Hermitian classical operator $G_C$,
\begin{align}
  \|\Dcal^\dagger(G_C)\|_\infty
  \le
  \|G_C\|_\infty
\end{align}
holds (see Appendix~\ref{app:contractivity}).

\subsubsection{Block decoders}

Let \(V=[m]\) be the set of classical variables and partition it into mutually disjoint blocks, \(V=\bigsqcup_{B\in\mathcal B}B\). For each block \(B\), we associate a classical register \(C_B\), a classical Hilbert space \(\mathcal H_{C_B}\), a quantum subsystem \(Q_B\), and a local POVM \(\{M_{s_B}^{(B)}\}_{s_B\in\{\pm1\}^B}\). The corresponding local decoder is denoted by \(\mathcal D_B\). For the block-structured decoders considered below, the overall decoder is constructed as
$\Dcal=\bigotimes_{B\in\Bcal}\Dcal_B$.

For any $A\subseteq V$, letting $A_B:=A\cap B$, the classical parity operator factorizes as \begin{align}
  S_{C,A}
  =
  \bigotimes_{B\in\Bcal}S_{C_B,A_B}.
  \label{eq:classical-block-factorization}
\end{align}
Since the decoder itself is a tensor product of the local decoders, its dual map factorizes accordingly: \begin{align}
  \Dcal^\dagger(S_{C,A})
  =
  \bigotimes_{B\in\Bcal}
  \Dcal_B^\dagger(S_{C_B,A_B}).
  \label{eq:decoder-block-factorization}
\end{align}
This factorization underlies both block-structured QRAO and the Fourier-component-preserving decoder designs developed below.

\subsection{Decoder-consistent Hamiltonian}
\label{sec:decoder-consistency}

We now give the main definition and theorem. Once a POVM decoder and a classical objective function are fixed, the Hermitian operator whose expectation value exactly reproduces the expected post-decoding objective value is uniquely given by the pullback of the classical objective operator under the decoder's dual map. We further show that, if only the ordering of quantum states is required to agree, this operator is unique up to a positive affine transformation.

\subsubsection{Definition and expectation-value identity}

\begin{definition}[Decoder-consistent Hamiltonian] 
\label{def:decoder-consistent-hamiltonian}
For a POVM decoder $\Dcal$ and an objective function $f$,
\begin{align}
H_{\Dcal}[f]:=\Dcal^\dagger(F_C)=\sum_{s\in\X}f(s)M_s
\label{eq:decoder-consistent-hamiltonian}
\end{align}
is called the decoder-consistent Hamiltonian.
\end{definition}
For any quantum state \(\rho\), the decoder-consistent Hamiltonian satisfies
\begin{align}
  \Tr(H_{\Dcal}[f]\rho)
  %&=\Tr\!\left(\Dcal^\dagger(F_C)\rho\right)
  %\nonumber\\
  &=\sum_{s\in\X}f(s)\Tr(M_s\rho)
  \\
  &=\E_{s\sim\Dcal(\rho)}[f(s)].
  \label{eq:decoded-expectation-derivation}
\end{align}
This identity is a direct consequence of the standard duality between Schrödinger- and Heisenberg-picture descriptions of a QC channel. Here we use it as a criterion for determining whether the operator used as the quantum optimization objective in a QRAO-type relaxation is consistent with the final decoder.
% \begin{newenglish}
% For a general POVM decoder, following the convention in quantum optimization, we refer to the Hermitian operator in Eq.~\eqref{eq:decoder-consistent-hamiltonian} as the Hamiltonian to be optimized.
% \end{newenglish}
%By definition, for any quantum state $\rho$, $\Tr(H_{\Dcal}[f]\rho)$ equals the post-decoding expected objective value. 
Therefore, if the goal is to directly maximize the post-decoding expected objective value, the natural operator for quantum optimization is $H_{\Dcal}[f]$.
In particular, if an eigenstate corresponding to the maximum eigenvalue of $H_{\Dcal}[f]$ can be prepared exactly, it maximizes the expected post-decoding objective value under the fixed decoder $\Dcal$.

\subsubsection{Order-preserving uniqueness}

The following theorem shows that the decoder-consistent Hamiltonian is essentially unique in the sense of preserving the ordering over the entire state space.

\begin{theorem}[Order-preserving uniqueness]
\label{thm:order-uniqueness}
Suppose $G:=H_{\Dcal}[f]$ is not proportional to the identity. For any Hermitian operator $H$, the following two statements are equivalent: 
\begin{enumerate}
  \item for all quantum states $\rho_1,\rho_2$, $\Tr(H\rho_1)>\Tr(H\rho_2)$ if and only if $\Tr(G\rho_1)>\Tr(G\rho_2)$;
  \item there exist $a>0$ and $b\in\mathbb R$ such that $H=aG+bI$.    
\end{enumerate}
\end{theorem}

The proof is given in Appendix~\ref{app:order-uniqueness-proof}. Consequently, if a Hermitian operator \(H\) is not a positive affine transformation of \(H_{\Dcal}[f]\), there exists a pair of states whose ordering under \(H\) does not agree with the ordering induced by the expected post-decoding objective value. Thus, the only objective operators that exactly reproduce the ordering induced by the expected post-decoding objective value over the full state space are the positive affine transformations of \(H_{\Dcal}[f]\). This uniqueness statement concerns the full set of density operators; agreement after restricting to a particular ansatz family is a separate question.

We finally note a basic property of the decoder-consistent Hamiltonian: its spectrum is confined to the range of the classical objective function.

\begin{proposition}[Spectral range]
\label{prop:spectral-range}
Let $f_{\min}:=\min_s f(s)$ and $f_{\max}:=\max_s f(s)$. For any POVM decoder, 
\begin{align}
  f_{\min}\Id
  \le
  H_{\Dcal}[f]
  \le
  f_{\max}\Id.
  \label{eq:spectral-range}
\end{align}
In particular, $\lambda_{\max}(H_{\Dcal}[f])\le f_{\max}$.
\end{proposition}

\begin{proof}
Since $f(s)-f_{\min}\ge0$ and $M_s\ge0$ for every $s$, we have $H_{\Dcal}[f]-f_{\min}I=\sum_s(f(s)-f_{\min})M_s\ge0$. Similarly, $f_{\max}I-H_{\Dcal}[f]=\sum_s(f_{\max}-f(s))M_s\ge0$.
\end{proof}

This proposition shows that the largest eigenvalue of the decoder-consistent Hamiltonian cannot exceed the classical optimum \(f_{\max}\). By decoder consistency, the expected objective value obtained by decoding a maximum-eigenvalue state is exactly \(\lambda_{\max}(H_{\Dcal}[f])\). Therefore, a lower bound on this eigenvalue that is close to \(f_{\max}\) directly guarantees that the expected post-decoding objective value is close to the classical optimum. Moreover, when the expectation value of \(H_{\Dcal}[f]\) is close to \(f_{\max}\), one can also obtain a probabilistic guarantee that the decoded distribution concentrates on optimal or near-optimal solutions; see Appendix~\ref{app:eigenvalue-closeness}.

\section{Pullback of Fourier components and existing QRAO-type relaxations}
\label{sec:parity-qrao}

We decompose the decoder-consistent Hamiltonian into the Fourier components of the classical objective function and use this language to reinterpret existing QRAO-type relaxations. The pullback formulation makes it possible to quantify a decoder by which Fourier components it preserves on the quantum side and to what extent.

\subsection{Pullback identity}

The decoder-consistent Hamiltonian is decomposed as 
\begin{align}
  H_{\Dcal}[f]
  &=
  \sum_{A\subseteq[m]}
  \widehat f_A\Dcal^\dagger(S_{C,A}).
  \label{eq:parity-pullback-identity}
\end{align}
We write 
\begin{align}
  O_A^{\Dcal}
  :=
  \Dcal^\dagger(S_{C,A})
  =
  \sum_{s\in\X}s^A M_s.
  \label{eq:parity-pullback-observable}
\end{align}
This identity allows us to track, component by component, the quantum operator to which each Fourier component is mapped by the decoder. Even for the same classical objective, different decoders generally give different $O_A^{\Dcal}$ and therefore different quantum objective Hamiltonians.

\subsection{Preservation rate}

For a Fourier component \(A\) with \(O_A^\Dcal\neq0\), we define
\begin{align}
  \kappa_A
  :=
  \|O_A^{\Dcal}\|_\infty,
  \qquad
  \widetilde O_A
  :=
  \frac{O_A^{\Dcal}}{\|O_A^{\Dcal}\|_\infty}.
  \label{eq:kappa-definition}
\end{align}
Thus,
\[
O_A^\Dcal=\kappa_A\widetilde O_A,
\qquad
\|\widetilde O_A\|_\infty=1.
\]
We refer to \(\kappa_A\) as the preservation rate of Fourier component \(A\), in the sense that it quantifies the operator-norm magnitude of its pullback under the decoder. The term is used only in this operator-norm sense; by itself, $\kappa_A$ does not imply independent recoverability of distinct Fourier components.
If \(O_A^\Dcal=0\), we set \(\kappa_A=0\) and may choose \(\widetilde O_A\) to be any fixed unit-norm Hermitian operator, since the product \(\kappa_A\widetilde O_A\) is always zero. For \(A=\emptyset\), \(O_\emptyset^\Dcal=I\), so we set \(\kappa_\emptyset=1\) and \(\widetilde O_\emptyset=I\).
Since the dual map of a POVM decoder is positive and unital, $0\le\kappa_A\le1$.

With this notation,
\begin{align}
  H_{\Dcal}[f]
  =
  \sum_{A\subseteq[m]}
  \widehat f_A\kappa_A\widetilde O_A.
  \label{eq:hamiltonian-with-kappa}
\end{align} 
In particular, a Fourier component with $\kappa_A=0$ makes no contribution to the post-decoding expectation for any quantum state.

\subsection{Reinterpretation of existing QRAO-type relaxations}
\label{sec:qrao-reinterpretation}

Using the language of decoder consistency and Fourier-component pullbacks, we reinterpret selected existing compression-based quantum relaxations in a unified manner. Our focus is not to rederive every detail of these methods, but to analyze the specific fixed decoders stated below and identify which Fourier components are pulled back by those decoders and with what preservation rates. The main examples are summarized in Table~\ref{tab:decoder-comparison}.

The present subsection concerns the structure of the decoder-induced pullbacks and their relation to the relaxation Hamiltonians. Performance bounds require the additional test-state signal rates \(\mu_A\) and the effective preservation profile \(\alpha_A=\kappa_A\mu_A\); these guarantees are developed afterward in Sec.~\ref{sec:test-state-performance}.

\begin{table*}[t]
\caption{
Representative decoders and existing relaxations in the framework of this paper.
Here, $\kappa_A$ denotes the preservation rate appearing in the pullback of a nonconstant Fourier component $A$ under the dual map of the decoder.
}
\label{tab:decoder-comparison}
\begin{ruledtabular}
\begin{tabular}{
p{0.15\textwidth}
p{0.25\textwidth}
p{0.17\textwidth}
p{0.15\textwidth}
p{0.20\textwidth}
}
Decoder / method
&
Main preserved components
&
Preservation rate $\kappa_A$
&
Decoder-consistent Hamiltonian
&
Interpretation / relation to conventional relaxation
\\
$(3,1)$-QRAC-type simultaneous-decoding POVM
&
Intra-block one-body components and inter-block two-body components
&
One-body: $1/\sqrt3$; inter-block two-body: $1/3$
&
Eq.~\eqref{eq:31-decoder-consistent-qubo}
&
Known positive affine equivalence for MaxCut; generally absent for mixed-degree QUBO under this decoder
\\
$(n,n-1)$-QRAC-type simultaneous-decoding POVM
&
Intra-block one-body components and inter-block two-body components
&
One-body: $\sqrt{\nu}$; inter-block two-body: $\nu$
&
Eq.~\eqref{eq:nnminus1-decoder-consistent-qubo}
&
Positive affine equivalence for the associated MaxCut construction; generally absent for mixed-degree QUBO under this decoder
\\
$(2,1)$ direct parity block
&
Prescribed local one- and two-body components
&
$s_1$: $a$; $s_2$: $b$; $s_1s_2$: $c$
&
Eq.~\eqref{eq:21-direct-parity-hamiltonian}
&
Depends on the decoder design
\\
POVM induced by Teramoto et al.'s magic rounding
&
$s_1$, $s_2$, and $s_1s_2$
&
Each component: $2/(3\sqrt3)$
&
Eq.~\eqref{eq:teramoto-decoder-consistent-hamiltonian}
&
Determined by the POVM induced by the measurement and classical post-processing
\end{tabular}
\end{ruledtabular}
\end{table*}

We briefly comment on direct encoding without compression, which provides a trivial limiting case of the present framework. If computational-basis measurement is used as the decoder, each Fourier component is pulled back directly as \(\Dcal^\dagger(S_{C,A})=\prod_{i\in A}Z_i\), so every component has preservation rate \(\kappa_A=1\). This provides a trivial no-compression example of the decoder-consistency framework.

\subsubsection{\texorpdfstring{$(3,1)$-QRAC simultaneous-decoding POVM}{(3,1)-QRAC simultaneous-decoding POVM}}

Consider a \((3,1)\)-QRAC-type block in which three classical variables are assigned to a single qubit~\cite{fuller2024}. For the three local positions \(a\in[3]\), let $\Gamma_1=X$, $\Gamma_2=Y$, and $\Gamma_3=Z$.
We define the simultaneous-decoding POVM as 
\begin{align}
  M_{s_B}^{(3,1)}
  =
  \frac18
  \left(
    \Id+\frac1{\sqrt3}(s_1X+s_2Y+s_3Z)
  \right)
  \label{eq:31-simultaneous-povm}
\end{align}
with $s_B\in\{\pm1\}^3$.
This simultaneous-decoding POVM outputs estimates of all three classical variables in the block and is the effective POVM induced by uniform magic-state rounding for the \((3,1)\)-QRAC~\cite{fuller2024}.

For each \(a\in[3]\), the one-body component is pulled back as
\[
\Dcal_B^\dagger(S_{C_B,\{a\}})=\frac{1}{\sqrt3}\Gamma_a,
\]
and hence has preservation rate \(1/\sqrt3\).

On the other hand, two-body components within the same block vanish. For example,
\[
\Dcal_B^\dagger(S_{C_B,\{1,2\}})=0.
\]
The sum vanishes because each term is odd in at least one independently summed sign. The other two-body components and the three-body component vanish similarly.
Under the graph-coloring constraint, the two variables associated with each interaction edge belong to different blocks. For a global variable \(i\), let \(a(i)\in[3]\) denote its local position within block \(B(i)\), and let \(P_i\) denote the corresponding Pauli operator acting on that block, i.e.,
\[P_i:=\Gamma_{a(i)}^{(B(i))}.\]
The pullback therefore factorizes as a tensor product, yielding
\[
\Dcal^\dagger(S_{C,\{i,j\}})
=\frac13 P_iP_j,
\qquad B(i)\neq B(j).
\]
Thus, inter-block two-body components are preserved with rate \(1/3\).

For a quadratic objective 
\begin{align}
  f(s)=c+\sum_i h_is_i+\sum_{(i,j)\in E}J_{ij}s_is_j
  \label{eq:general-qubo}
\end{align}
satisfying the coloring constraint, the decoder-consistent Hamiltonian is
\begin{align}
  H_{\Dcal}^{(3,1)}[f]
  =
  c\Id
  +
  \frac1{\sqrt3}\sum_i h_iP_i
  +
  \frac13\sum_{(i,j)\in E}J_{ij}P_iP_j.
  \label{eq:31-decoder-consistent-qubo}
\end{align}

\subsubsection{\texorpdfstring{Extension to the $(n,n-1)$-QRAC simultaneous-decoding POVM}{Extension to the (n,n-1)-QRAC simultaneous-decoding POVM}}

The \((n,n-1)\)-QRAC-type construction encodes \(n\) classical bits into \(n-1\) qubits. In low dimensions, the \((2,1)\)-QRAC is a standard example of a one-qubit QRAC~\cite{ambainis2002,imamichi2018}, while a \((3,2)\)-QRAC was identified through numerical search~\cite{imamichi2018}. Related constructions have also appeared in the QRAO literature~\cite{teramoto2023}. The explicit construction and coefficient relations used here follow Ref.~\cite{suzuki2026nnminus1qrac}; see also Ref.~\cite{kondo2026rac} for an alternative formulation. Detailed code states and the derivation of the blockwise simultaneous-decoding POVM are provided in Appendix~\ref{app:nnminus1-qrao}. In the main text, we use only the relations needed under the coloring constraint, where interacting variables lie in different blocks.

Let \(\nu:=(n-1)/n\), and let \(O_i\) denote the unit-norm Hermitian operator associated with variable \(i\) in the \((n,n-1)\)-QRAC construction. Then, under simultaneous decoding, the one-body component has preservation rate $\sqrt\nu$, while a two-body component connecting different blocks has preservation rate $\nu$. Specifically,
\begin{align}
  \Dcal^\dagger(S_{C,\{i\}})&=\sqrt\nu O_i,
  \\
  \Dcal^\dagger(S_{C,\{i,j\}})&=\nu O_iO_j,
  \qquad B(i)\neq B(j).
  \label{eq:nnminus1-basic-pullback}
\end{align}
Under the coloring constraint, the decoder-consistent Hamiltonian is 
\begin{align}
  H_{\Dcal}^{(n,n-1)}[f]
  =
  c\Id
  +
  \sqrt\nu\sum_i h_iO_i
  +
  \nu\sum_{(i,j)\in E}J_{ij}O_iO_j.
  \label{eq:nnminus1-decoder-consistent-qubo}
\end{align}

\subsubsection{Affine equivalence for MaxCut}

For an unweighted graph $G=(V,E)$, we write the MaxCut objective as
\begin{align}
  \Cut(s)
  =
  \frac{|E|}{2}
  -
  \frac12\sum_{(i,j)\in E}s_is_j.
  \label{eq:maxcut-objective}
\end{align}
For the \((3,1)\)-QRAC simultaneous-decoding POVM under the coloring constraint, all nonconstant MaxCut terms have the same preservation rate. The decoder-consistent Hamiltonian is therefore
\begin{align}
  H_{\Dcal}^{(3,1)}[\Cut]
  =
  \frac{|E|}{2}\Id
  -
  \frac16\sum_{(i,j)\in E}P_iP_j.
  \label{eq:31-decoder-consistent-maxcut}
\end{align}
If the conventional relaxation Hamiltonian, which compensates for the
attenuation on the code states, is written as
\begin{align}
H_{\mathrm{relax}}^{(3,1)}
&:=
\frac12\sum_{(i,j)\in E}
\left(I-3P_iP_j\right)
\nonumber\\
&=
\frac{|E|}{2}I
-\frac32\sum_{(i,j)\in E}P_iP_j,
\label{eq:31-relax-maxcut}
\end{align}
then
\begin{align}
H_{\Dcal}^{(3,1)}[\Cut]-\frac{|E|}{2}I
=
\frac19
\left(
H_{\mathrm{relax}}^{(3,1)}
-\frac{|E|}{2}I
\right).
\label{eq:31-maxcut-affine}
\end{align}

Similarly, for $(n,n-1)$-QRAC we write the natural code-state-compensated MaxCut relaxation associated with this encoding as 
\begin{align}
  H_{\mathrm{relax}}^{(n,n-1)}
  &:=
  \frac12\sum_{(i,j)\in E}
  \left(
    \Id-\nu^{-1}O_iO_j
  \right)
  \nonumber\\
  &=
  \frac{|E|}{2}\Id
  -
  \frac{1}{2\nu}\sum_{(i,j)\in E}O_iO_j
  \label{eq:nnminus1-relax-maxcut}
\end{align}
This comparison Hamiltonian compensates for the attenuation of the two-body code-state expectation by the factor $1/\nu$.
In the $(n,n-1)$ case specifically, 
\begin{align}
  H_{\Dcal}^{(n,n-1)}[\Cut]-\frac{|E|}{2}\Id
  =
  \nu^2
  \left(
    H_{\mathrm{relax}}^{(n,n-1)}-\frac{|E|}{2}\Id
  \right)
  \label{eq:nnminus1-maxcut-affine}
\end{align}
Thus, within these constructions, a MaxCut objective whose nonconstant terms all acquire the same preservation factor yields a positive affine relation between the code-state-compensated relaxation and the decoder-consistent Hamiltonian. By Theorem~\ref{thm:order-uniqueness}, the two operators then induce the same ordering of states. For the standard $(3,1)$ MaxCut magic-state rounding, this recovers the result of Fuller et al. that the rounded average cut is linearly related to the relaxed energy and is monotone in that energy~\cite{fuller2024}; Eq.~\eqref{eq:31-maxcut-affine} expresses that known relation directly as a positive affine equivalence with the POVM pullback.

The RQRAO unification theorem of Ref.~\cite{kondo2025} establishes a positive affine relation between the expectation value of the standard QRAO/RQRAO Hamiltonian and the expected cut value after decoding. It can therefore be interpreted as a concrete instance of Hamiltonian--decoder consistency for the expected MaxCut value.

More explicitly, for the random magic measurement of Ref.~\cite{kondo2025}, let $P_m(b;\rho)$ be the output distribution and $\Cut(b)$ the corresponding MaxCut value. The Hamiltonian \(H_m\) used in the standard QRAO/RQRAO formulation satisfies 
\begin{align}
  \Tr(H_m\rho)
  =
  a\,\E_{b\sim P_m(\cdot;\rho)}[\Cut(b)]
  +
  b_0
\end{align}
for some $a>0$ and $b_0\in\mathbb R$.
Viewing the random magic measurement as a single POVM decoder $\Dcal_m^{\mathrm{magic}}$, decoder consistency gives
\begin{align}
  \E_{b\sim P_m(\cdot;\rho)}[\Cut(b)]
  =
  \Tr\!\left(H_{\Dcal_m^{\mathrm{magic}}}[\Cut]\rho\right).
\end{align}
Since the relation above holds for arbitrary \(\rho\), $H_m=aH_{\Dcal_m^{\mathrm{magic}}}[\Cut]+b_0I$. Thus, the Hamiltonian appearing in the unification theorem is a positive affine transformation of the decoder-consistent Hamiltonian associated with the random magic measurement. Because \(a>0\), the two Hamiltonians induce the same ordering of states, and optimizing the expectation value of \(H_m\) is equivalent to optimizing the expected decoded cut value.

\subsubsection{Non-affine relation for general QUBO objectives}

For quadratic objective functions that contain both one-body and two-body terms, the relation depends on the decoder and on the scaling used in the relaxation. In the $(3,1)$-QRAC case, the decoder-consistent Hamiltonian for the simultaneous-decoding POVM above is given by Eq.~\eqref{eq:31-decoder-consistent-qubo}.

As a concrete comparison object, consider the degree-dependent code-state-compensated Hamiltonian
\begin{align}
H_{\mathrm{relax}}^{(3,1)}[f]
=
 cI+\sqrt3\sum_i h_iP_i
 +3\sum_{(i,j)\in E}J_{ij}P_iP_j.
\label{eq:31-relax-qubo}
\end{align}
This scaling extends the MaxCut coefficient compensation of Ref.~\cite{fuller2024} and is used in general-Ising QRAO formulations and in the Qiskit reference encoding~\cite{matsuyama2024internal,qiskitOptimizationQRAOEncoding2023}. Related applications to constrained and QUBO formulations are discussed in Refs.~\cite{sharma2024knapsack,sharma2025comparative}.

Let
\[
A_1:=\sum_i h_iP_i,
\qquad
A_2:=\sum_{(i,j)\in E}J_{ij}P_iP_j.
\]
Under the coloring constraint, $A_1$ and $A_2$ lie in distinct Pauli-weight sectors and are Hilbert--Schmidt orthogonal. The centered nonconstant parts are
\begin{align}
H_{\Dcal}^{(3,1)}[f]-cI
&=
\frac1{\sqrt3}A_1+\frac13 A_2,
\\
H_{\mathrm{relax}}^{(3,1)}[f]-cI
&=
\sqrt3 A_1+3A_2.
\end{align}
If $A_1\neq0$ and $A_2\neq0$, a positive affine relation would require the same scale factor to equal both $1/3$ from the one-body sector and $1/9$ from the two-body sector, which is impossible. In the $(n,n-1)$ construction, the analogous code-state-compensated comparison has one-body and inter-block two-body coefficients $\nu^{-1/2}$ and $\nu^{-1}$, whereas the simultaneous-decoding POVM gives $\sqrt\nu$ and $\nu$. The corresponding ratios are therefore $\nu$ and $\nu^2$, and the same conclusion follows whenever the aggregate one-body and inter-block two-body operator sectors are linearly independent and both nonzero.

Consequently, under these stated decoder and sector conditions, the degree-scaled relaxation and the decoder-consistent Hamiltonian are not related by a positive affine transformation and can induce different orderings over the full state space. This mismatch does not imply that the degree-scaled relaxation is invalid as a code-state-reproducing, lower-bounding, or heuristic surrogate objective. It means only that its expectation value cannot, in general, be identified with the expected decoded objective value for the specific simultaneous-decoding POVM considered here. The Qiskit reference encoding implements the degree-dependent scaling in Eq.~\eqref{eq:31-relax-qubo}, but the encoding prescription by itself does not fix a rounding scheme~\cite{qiskitOptimizationQRAOEncoding2023}. Accordingly, the present comparison does not assume that every implementation uses the simultaneous-decoding POVM above. When numerical results are interpreted specifically as improvements in the expected post-decoding objective value, the Hamiltonian--decoder relation must therefore be checked for the decoder actually used.

\subsection{Performance bounds from test states}
\label{sec:test-state-performance}

The preservation rate \(\kappa_A\) quantifies how much of each Fourier component is retained on the quantum side by the dual map of the decoder. To lower-bound the maximum eigenvalue of \(H_{\Dcal}[f]\) via the variational principle, however, we must also quantify the expectation values of the corresponding normalized quantum operators on a suitable family of comparison states.
The test states used here are not states that the algorithm is required to prepare; they are comparison states used only as witnesses in the variational lower bound on $\lambda_{\max}(H_{\Dcal}[f])$.

Consider a family of test states \(\{\tau_s\}_{s\in\X}\) such that, for each nonzero Fourier component $A$ of interest,
\[
\Tr(\widetilde O_A\tau_s)=\mu_A s^A,
\qquad \mu_A\in[0,1].
\]
We call \(\mu_A\) the signal rate of component \(A\).
We define $\alpha_A:=\kappa_A\mu_A$ as the effective preservation rate of component $A$. Then $\Tr(\Dcal^\dagger(S_{C,A})\tau_s)=\alpha_A s^A$.

The following result provides a lower bound on \(\lambda_{\max}(H_{\Dcal}[f])\) in terms of the effective preservation rates, assuming that such test states exist. Their existence is not guaranteed for an arbitrary decoder, but they can be constructed explicitly for the representative blocks considered in this paper. A precise definition of the test states and their explicit constructions are given in Appendix~\ref{app:test-states}.

For the objective $f$ considered below, we require this relation for all nonconstant $A\in\operatorname{supp}_{\mathrm{F}}(f)$ with $O_A^\Dcal\neq0$. Write the objective as
\[
f(s)=C+\sum_{A\neq\emptyset}\widehat f_A s^A,
\qquad
C:=\widehat f_{\emptyset}.
\]
For any classical configuration $s$, the variational principle gives 
\begin{align}
  \lambda_{\max}(H_{\Dcal}[f])
  &\ge
  \Tr\!\left(H_{\Dcal}[f]\tau_{s}\right)
  \nonumber\\
  &=
  C+
  \sum_{A\neq\emptyset}\widehat f_A\alpha_A(s)^A.
  \label{eq:general-test-state-bound}
\end{align}
Taking $s=s^*$ to be a classical optimal solution yields a lower bound on \(\lambda_{\max}(H_{\Dcal}[f])\).
In particular, if every nonconstant component appearing in the objective
has the same effective preservation rate $\alpha_A=\alpha$, then
\begin{align}
\lambda_{\max}(H_{\Dcal}[f])
\ge
C+\alpha(f_{\max}-C).
\label{eq:uniform-alpha-bound}
\end{align}
Thus, when the effective preservation rate is uniformly \(\alpha\), the guaranteed gain above the baseline \(C\) is reduced to an \(\alpha\) fraction of the classical optimal gain \(f_{\max}-C\).

For the MaxCut objective in Eq.~\eqref{eq:maxcut-objective}, the constant term is $C=|E|/2$ and all nonconstant components are two-body spin monomials. When all edges have the same effective preservation rate, as in standard coloring-based QRAO constructions~\cite{fuller2024}, the general bound takes a particularly simple form.
Therefore, if all edges have the same effective preservation rate $\alpha$ and $\rho_{\max}$ is a maximum-eigenvalue state of $H_{\Dcal}[\Cut]$, decoder consistency gives
\begin{align}
  \E_{s\sim\Dcal(\rho_{\max})}[\Cut(s)]
  &=
  \Tr(H_{\Dcal}[\Cut]\rho_{\max})
  \nonumber\\
  %&=
  %\lambda_{\max}(H_{\Dcal}[\Cut])
  %\nonumber\\
  &\ge
  \frac{|E|}{2}
  +
  \alpha\left(\Cut_{\max}-\frac{|E|}{2}\right).
  \label{eq:maxcut-alpha-bound-decoded}
\end{align}

When different components have different effective preservation rates, a component-dependent guarantee is required. Appendix~\ref{app:maxcut-two-rate-bound} derives a more detailed MaxCut lower bound for the case in which intra-block and inter-block edges have different effective preservation rates.

\section{General construction of Fourier-component-preserving decoders}
\label{sec:decoder-design}

Up to this point, we have used the pullback to analyze a given decoder. We now reverse this perspective: starting from the Fourier components that we wish to preserve, we construct a decoder whose pullback realizes them with prescribed strengths. Concretely, we use the preservation rates \(\kappa_A\) as design variables for POVM decoder blocks. Starting from the Fourier decomposition of the objective function, we choose in each block a set of local components to be preserved in the expected decoded objective and assign them to Hermitian operators on the quantum side. We then choose their preservation rates subject to positivity of the POVM elements. We refer to the resulting trade-off between the number of preserved components and their preservation rates as the positivity budget.

In practical implementations, the Hilbert-space dimension required to accommodate the assigned operator family and the cost of implementing the measurement impose additional constraints. Here, we focus first on the fundamental trade-off arising from POVM positivity.

\subsection{Block-local Fourier-component decoder}

Let $\Kcal_B\subseteq 2^B\setminus\{\emptyset\}$ be the set of nontrivial Fourier components we wish to preserve in block $B$.
%For example, choosing $\Kcal_B=\{\{1\},\{2\},\{1,2\}\}$ for $B=\{1,2\}$ means preserving the one-body components $s_1,s_2$ and the two-body component $s_1s_2$.
For each \(A\in\Kcal_B\), we assign a Hermitian operator \(\Gamma_A^{(B)}\) on \(Q_B\), normalized as
\begin{align}
  \|\Gamma_A^{(B)}\|_\infty=1.
  \label{eq:gamma-normalization}
\end{align}
For preservation rates $\kappa_A^{(B)}>0$, consider the following candidate POVM elements, indexed by $s_B\in\{\pm1\}^B$:
\begin{align}
  M_{s_B}^{(B)}
  =
  \frac1{2^{|B|}}
  \left(
    \Id_{Q_B}
    +
    \sum_{A\in\Kcal_B}
    \kappa_A^{(B)}s^A\Gamma_A^{(B)}
  \right).
  \label{eq:parity-preserving-block-povm}
\end{align}
Because every \(A\in\Kcal_B\) is nonempty, the nonconstant parity terms vanish when summed over all outcomes, and hence \(\sum_{s_B}M_{s_B}^{(B)}=I_{Q_B}\). It therefore remains only to ensure positivity.

A sufficient condition for positivity is \begin{align}
  \left\|
    \sum_{A\in\Kcal_B}
    \kappa_A^{(B)}s^A\Gamma_A^{(B)}
  \right\|_\infty
  \le1,\qquad \forall s_B\in\{\pm1\}^B. 
  \label{eq:positivity-budget-general}
\end{align}
In particular, if $\Gamma_A^{(B)}$ are mutually anticommuting Pauli operators, the sufficient condition simplifies to
\begin{align}
  \sum_{A\in\Kcal_B}\left(\kappa_A^{(B)}\right)^2\le1.
  \label{eq:l2-positivity-budget}
\end{align}
A derivation is given in Appendix~\ref{app:positivity-and-pullback}.
For mutually anticommuting Pauli assignments, Eq.~\eqref{eq:l2-positivity-budget} quantifies the trade-off between the number \(K_B:=|\Kcal_B|\) of preserved components and their preservation rates. The available Hilbert-space dimension further constrains how large such an anticommuting operator family can be. When $Q_B$ is a $q_B$-qubit register, $\dim\Hcal_{Q_B}=2^{q_B}$ and the compression ratio is $|B|/q_B$. Within this construction, preserving more components while keeping \(q_B\) small requires them to share the positivity budget in Eq.~\eqref{eq:l2-positivity-budget}. For example, if three Fourier components are assigned symmetrically to the three mutually anticommuting Pauli axes of a single qubit, then \(K_B=3\), and equal preservation rates satisfy \(\kappa_A^{(B)}\le 1/\sqrt3\).

\subsection{Pullback of preserved components and the coverage condition}

We now verify that, for the POVM decoder block constructed above, exactly the components in the preserved set \(\Kcal_B\) have nonzero pullbacks. We then define a coverage condition ensuring that every term in the objective survives in the decoder-consistent Hamiltonian of a product decoder.

\begin{proposition}[Pullback of Fourier components in the block]
\label{prop:block-parity-pullback}
For the block POVM in Eq.~\eqref{eq:parity-preserving-block-povm}, parity orthogonality gives
\begin{align}
  \Dcal_B^\dagger(S_{C_B,A'})
  =
  \begin{cases}
    \Id_{Q_B}, & A'=\emptyset,\\[1mm]
    \kappa_{A'}^{(B)}\Gamma_{A'}^{(B)}, & A'\in\Kcal_B,\\[1mm]
    0, & A'\notin\Kcal_B,\ A'\neq\emptyset
  \end{cases}.
  \label{eq:block-parity-pullback-formula}
\end{align}
\end{proposition}
\begin{proof}
Substituting the POVM definition into $\Dcal_B^\dagger(S_{C_B,A'})=\sum_{s_B}s^{A'}M_{s_B}^{(B)}$ and using the orthogonality relations of spin monomials,
\begin{align}
  2^{-|B|}\sum_{s_B}s^{A'}=\delta_{A',\emptyset},
  \qquad
  2^{-|B|}\sum_{s_B}s^{A'}s^A=\delta_{A',A}
\end{align}
immediately yields Eq.~\eqref{eq:block-parity-pullback-formula}.
\end{proof}
Thus, among the nonconstant local Fourier components, exactly those included in the preserved set have nonzero pullbacks.

For the full product decoder, we say that a Fourier component \(A\subseteq V\) is \emph{covered} by the block partition and the family of preserved sets \(\{\Kcal_B\}\) if \(A_B:=A\cap B\in\Kcal_B\cup\{\emptyset\}\) for every block \(B\). If every nonconstant Fourier component of \(f\) with a nonzero coefficient is covered, we call the decoder term-complete for \(f\).
For example, if a two-body term $J_{12}s_1s_2$ appears within a block $B$ but the corresponding component $\{1,2\}$ is not included in $\Kcal_B$, Proposition~\ref{prop:block-parity-pullback} gives $\Dcal_B^\dagger(S_{C_B,\{1,2\}})=0$. Consequently, that two-body term does not appear in the decoder-consistent Hamiltonian.
% \begin{newenglish}
% As a concrete example, consider $(3,1)$-QRAO-type simultaneous decoding with $V=\{1,\ldots,6\}$ partitioned into $B_1=\{1,3,5\}$ and $B_2=\{2,4,6\}$, where each preserved set contains only one-body components. The colored quadratic objective $f(s)=J_{12}s_1s_2+J_{34}s_3s_4+J_{56}s_5s_6$ is term-complete because every edge decomposes into one preserved one-body component in each block. By contrast, a term $J_{13}s_1s_3$ is not covered because $\{1,3\}\notin\Kcal_{B_1}$.
% \end{newenglish}
Thus, term-completeness is precisely the condition that every nonconstant Fourier component appearing in the objective with a nonzero coefficient has a nonzero pullback under the decoder.

By Proposition~\ref{prop:block-parity-pullback} and the block factorization in Eq.~\eqref{eq:decoder-block-factorization}, a covered component \(A\) satisfies
\begin{align}
  \Dcal^\dagger(S_{C,A})
  =
  \left(
    \prod_{B:A_B\neq\emptyset}
    \kappa_{A_B}^{(B)}
  \right)
  \bigotimes_{B:A_B\neq\emptyset}
  \Gamma_{A_B}^{(B)},
  \label{eq:covered-global-pullback}
\end{align}
where identity operators on blocks with \(A_B=\emptyset\) are left implicit.
If, instead, \(A_B\notin\Kcal_B\cup\{\emptyset\}\) for at least one block \(B\), then
\begin{align}
  \Dcal^\dagger(S_{C,A})=0.
  \label{eq:uncovered-global-pullback}
\end{align}
Accordingly, the decoder-consistent Hamiltonian can be written as
\begin{widetext}
\[H_{\Dcal}[f]
=
\widehat f_{\emptyset}I
+
\sum_{\substack{
A\in\operatorname{supp}_{\mathrm F}(f)\setminus\{\emptyset\}:\\
A\ {\rm covered}
}}
\widehat f_A
\left(
\prod_{B:A_B\neq\emptyset}
\kappa_{A_B}^{(B)}
\right)
\bigotimes_{B:A_B\neq\emptyset}
\Gamma_{A_B}^{(B)},\]
\end{widetext}
where identity operators on blocks with \(A_B=\emptyset\) are left implicit.

For a covered component $A$, the global preservation rate is therefore $\kappa_A=\prod_{B:A_B\neq\emptyset}\kappa_{A_B}^{(B)}$. An uncovered component can equivalently be regarded as having preservation rate zero, in which case its contribution disappears from the decoder-consistent Hamiltonian. Thus, term-completeness ensures that every nonconstant Fourier component with a nonzero coefficient in the objective has a nonzero pullback contribution to the decoder-consistent Hamiltonian.

\section{Examples}
\label{sec:examples-implications}

This section presents representative decoder blocks and discusses their analytical and practical implications. We focus on the central examples in the main text and defer the randomized Pauli implementation to Appendix~\ref{app:randomized-pauli-implementation}.

Section~\ref{sec:qrao-reinterpretation} classified fixed decoders through their Fourier preservation profiles. Here we make the two directions of the framework explicit. The $(2,1)$ direct parity block instantiates the inverse-design principle of Section~\ref{sec:decoder-design}: it starts from desired one- and two-body pullbacks and constructs a POVM that realizes them. We then derive the effective POVM induced by the magic rounding of Teramoto et al. from its actual measurement and classical post-processing. Although the two examples use the same formal assignment of spin monomials to Pauli operators, their preservation coefficients are determined by the corresponding POVMs rather than by that formal assignment alone.

\subsection{\texorpdfstring{$(2,1)$ direct parity block}{(2,1) direct parity block}}

Consider a block that compresses two classical variables \(s_1,s_2\) into a single qubit while preserving both one-body components \(s_1,s_2\) and the two-body component \(s_1s_2\). We take \(B=\{1,2\}\) and \(\mathcal K_B=\{\{1\},\{2\},\{1,2\}\}\) as the preserved set, and choose \(\Gamma_{\{1\}}=X\), \(\Gamma_{\{2\}}=Y\), and \(\Gamma_{\{1,2\}}=Z\).
The corresponding four-outcome POVM is
\begin{align}
  M_{s_1,s_2}^{(2,1)}
  =
  \frac14
  \left(
    \Id+a s_1X+b s_2Y+c s_1s_2Z
  \right),
  \quad
  a,b,c>0.
  \label{eq:21-direct-parity-povm}
\end{align}
Since \(X,Y,Z\) are mutually anticommuting, this is a valid POVM whenever \(a^2+b^2+c^2\le1\).

The corresponding pullbacks are
\begin{align}
  \Dcal_B^\dagger(S_{C_B,\{1\}})=aX,
  \\
  \Dcal_B^\dagger(S_{C_B,\{2\}})=bY,
  \\
  \Dcal_B^\dagger(S_{C_B,\{1,2\}})=cZ.
  \label{eq:21-direct-parity-pullback}
\end{align}
Hence, for the local objective $f_B(s_1,s_2)=h_1s_1+h_2s_2+J_{12}s_1s_2$, the decoder-consistent Hamiltonian is \begin{align}
  H_{\Dcal_B}[f_B]
  =
  ah_1X+bh_2Y+cJ_{12}Z.
  \label{eq:21-direct-parity-hamiltonian}
\end{align}
The symmetric choice \(a=b=c=1/\sqrt3\) gives the single-qubit tetrahedral SIC-POVM~\cite{renes2004sic}.

When the symmetric direct parity block is applied to MaxCut, two different effective preservation rates arise for intra-block and inter-block edges. A two-rate lower bound can then be derived for the maximum eigenvalue of the decoder-consistent Hamiltonian and for the expected cut value obtained by decoding its maximum-eigenvalue state. Because it separately accounts for the contribution of intra-block edges, this bound can be tighter than the single-rate bound in Eq.~\eqref{eq:maxcut-alpha-bound-decoded}. The explicit formula and its derivation are given in Appendix~\ref{app:maxcut-two-rate-bound}.

\subsection{\texorpdfstring{POVM induced by the magic rounding of Teramoto et al.}{POVM induced by the magic rounding of Teramoto et al.}}

The magic-rounding procedure of Teramoto et al.~\cite{teramoto2023}, when its random measurement choice and classical post-processing are combined into one effective measurement, provides an example of a POVM decoder that retains a two-body spin monomial within the same block.
The symmetric direct parity block and this induced POVM use the same formal assignment \(s_1\mapsto X\), \(s_2\mapsto Y\), and \(s_1s_2\mapsto Z\). The difference is that the former realizes this assignment directly as a four-outcome POVM, whereas the latter realizes it as the POVM induced by randomly selecting a magic basis, performing a binary projective measurement, and applying classical post-processing.

For each output \(s=(s_1,s_2)\), define the tetrahedral direction \(\vec n_s=(s_1,s_2,s_1s_2)/\sqrt3\). Combining the random magic-basis selection, binary projective measurement, and classical post-processing into a single POVM gives the induced effect
\begin{align}
  M_{s_1,s_2}^{\mathrm T}
  =
  \frac14
  \left[
    \Id+
    \frac{2}{3\sqrt3}
    (s_1X+s_2Y+s_1s_2Z)
  \right].
  \label{eq:teramoto-induced-povm}
\end{align}
A detailed derivation is given in Appendix~\ref{app:teramoto-derivation}.

Therefore, for this POVM decoder, $\Dcal_{\mathrm T}^\dagger(S_{C,\{1\}})=\kappa_{\mathrm T}X$, $\Dcal_{\mathrm T}^\dagger(S_{C,\{2\}})=\kappa_{\mathrm T}Y$, and $\Dcal_{\mathrm T}^\dagger(S_{C,\{1,2\}})=\kappa_{\mathrm T}Z$, where $\kappa_{\mathrm T}=2/(3\sqrt3)$.
For the same local objective $f_B(s_1,s_2)$, the decoder-consistent Hamiltonian is 
\begin{align}
H_{\Dcal_{\mathrm T}}[f_B]
=
\kappa_{\mathrm T}
(h_1X+h_2Y+J_{12}Z).
\label{eq:teramoto-decoder-consistent-hamiltonian}
\end{align}
Importantly, this common coefficient is determined by the actual POVM induced by the measurement and classical post-processing, rather than by the formal spin-to-operator assignment alone. This comparison shows that the same decoder-first pullback viewpoint can be used both to design new decoders, such as the direct parity block, and to reinterpret existing procedures, such as the magic rounding of Teramoto et al.

\section{Conclusion}
\label{sec:conclusion}

In this paper, we developed a decoder-first framework for characterizing the relation between a fixed POVM decoder and the quantum objective used in compression-based quantum optimization. For a classical objective function \(f\), the pullback
\[
H_{\Dcal}[f]=\Dcal^\dagger(F_C)=\sum_s f(s)M_s
\]
exactly represents the expected post-decoding objective value. Furthermore, in the nontrivial case, its positive affine class is the unique class of Hermitian objectives that induces the same strict ordering over the full quantum state space.

The Fourier-pullback analysis shows that decoder consistency is a structural property of the decoder--objective pair rather than an automatic consequence of code-state reproduction. The QRAO examples illustrate both sides of this distinction: for MaxCut, a common attenuation factor across all relevant components recovers the known positive affine relation, whereas for mixed-degree QUBO objectives, component-dependent attenuation can destroy that relation under the simultaneous-decoding POVMs analyzed here. The same pullback representation can also be used constructively, by selecting which Fourier components a decoder should preserve and allocating their strengths subject to POVM positivity.

Practical use of the framework additionally depends on the reachability of favorable states and on the measurement or sampling cost of the chosen decoder and objective under finite-shot and noisy conditions. These implementation questions are problem- and hardware-dependent.

Overall, the decoder-first viewpoint provides both a criterion for determining when a quantum objective is consistent with a prescribed decoder and a constructive language for designing decoders around the structure of the classical objective.

\appendix

\section{Technical preliminaries: POVMs and basic lemmas}
\label{app:technical-preliminaries}

\subsection{Contractivity of the dual decoder map}
\label{app:contractivity}

\begin{lemma}
Consider the quantum-to-classical channel \(\Dcal\) determined by a POVM \(\{M_s\}_{s\in\X}\). On the algebra of diagonal classical operators, the dual map \(\Dcal^\dagger\) is positive and unital. Therefore, for any diagonal Hermitian operator \(A\) on the classical register,
\[
\|\Dcal^\dagger(A)\|_\infty\le \|A\|_\infty.
\]
\end{lemma}

\begin{proof}
For a diagonal classical operator
\[
A=\sum_s a_s\ket{s}\bra{s},
\]the dual map is
\[
\Dcal^\dagger(A)=\sum_s a_sM_s.
\]If \(A\ge0\), then \(a_s\ge0\) for all \(s\), and hence \(\Dcal^\dagger(A)\ge0\). Moreover, the POVM completeness relation gives
\[
\Dcal^\dagger(I_C)=\sum_sM_s=I_Q,
\]so \(\Dcal^\dagger\) is unital. Since
\[
-\|A\|_\infty I_C
\le A\le
\|A\|_\infty I_C,
\]positivity and unitality imply
\[
-\|A\|_\infty I_Q
\le
\Dcal^\dagger(A)
\le
\|A\|_\infty I_Q.
\]Therefore,
\[
\|\Dcal^\dagger(A)\|_\infty
\le
\|A\|_\infty.
\]
\end{proof}
\subsection{Proof of the order-uniqueness theorem}
\label{app:order-uniqueness-proof}

\begin{proof}
First assume condition 2. If $H=aG+bI$ and $a>0$, then for any state $\rho$, $\Tr(H\rho)=a\Tr(G\rho)+b$, which immediately implies condition 1. We next assume condition 1 and prove condition 2.
Let \(V\) denote the real vector space of traceless Hermitian operators. Any difference of two quantum states belongs to \(V\). Conversely, for any nonzero \(\Delta\in V\), write its positive-negative decomposition as \(\Delta=\Delta_+-\Delta_-\). Since \(\Tr\Delta=0\), we have \(\Tr\Delta_+=\Tr\Delta_-=:c>0\). Hence, defining \(\rho_\pm=\Delta_\pm/c\), we obtain \(\Delta=c(\rho_+-\rho_-)\). Therefore, condition 1 is equivalent to
\begin{align}
\Tr(H\Delta)>0
\iff
\Tr(G\Delta)>0,
\qquad
\forall \Delta\in V.
\label{eq:order-functional-sign}
\end{align}
% \begin{pdfenglish}
% The difference of quantum states $\Delta=\rho_1-\rho_2$ is a traceless Hermitian operator. Conversely, any nonzero traceless Hermitian operator can be written as $\Delta=c(\rho_+-\rho_-)$ for some $c>0$ and two quantum states. Therefore condition 1 is equivalent to $\Tr(H\Delta)>0$ if and only if $\Tr(G\Delta)>0$ for every $\Delta$ in the real vector space $V$ of traceless Hermitian operators.
% \end{pdfenglish}

Define the real linear functionals \(L_H(\Delta)=\Tr(H\Delta)\) and \(L_G(\Delta)=\Tr(G\Delta)\) on \(V\). Since \(G\) is not a scalar multiple of the identity, \(L_G\neq0\). Equation~\eqref{eq:order-functional-sign} implies that \(L_H\) and \(L_G\) have the same positive region. In particular, their kernels must coincide. Indeed, if \(L_G(\Delta_0)=0\) but \(L_H(\Delta_0)\neq0\), then either \(\Delta_0\) or \(-\Delta_0\) has positive \(L_H\) value while its \(L_G\) value remains zero, contradicting Eq.~\eqref{eq:order-functional-sign}. The reverse inclusion follows in the same way.
Nonzero linear functionals with the same kernel are proportional, so \(L_H=aL_G\) for some \(a\neq0\). Choosing any \(\Delta\in V\) with \(L_G(\Delta)>0\) and using Eq.~\eqref{eq:order-functional-sign} shows that \(a>0\). Hence
\[
\Tr[(H-aG)\Delta]=0
\]for every \(\Delta\in V\). Since the Hilbert--Schmidt orthogonal complement of the traceless Hermitian subspace is \(\operatorname{span}_{\mathbb R}\{I\}\), there exists \(b\in\mathbb R\) such that
\[
H-aG=bI.
\]Therefore, \(H=aG+bI\), proving condition 2.
\end{proof}

\section{Additional guarantees for decoded solutions}
\label{app:guarantees-and-utilities}

\subsection{Near-optimal expectation and decoded optimal-solution probability}
\label{app:eigenvalue-closeness}

We show that when the expectation value of the decoder-consistent Hamiltonian is close to the classical optimum, one obtains a quantitative lower bound on the probability of decoding an optimal solution. This follows because the expectation value of \(H_{\Dcal}[f]\) is not merely a surrogate energy but exactly the expected post-decoding objective value.
This proposition does not provide a procedure for directly maximizing the probability of obtaining an optimal solution. Rather, assuming that the expected post-decoding objective value is already close to \(f_{\max}\), it quantifies the resulting concentration of probability on the optimal or near-optimal solution set.
\begin{proposition}[Near-optimal expectation and decoded optimal-solution probability]\label{prop:eigenvalue-closeness-optimal-probability}
Let $f_{\max}:=\max_s f(s)$ and $\X_{\max}:=\{s:f(s)=f_{\max}\}$. Unless all configurations are optimal, define the optimality gap $\Delta_\star:=f_{\max}-\max_{s\notin\X_{\max}}f(s)>0$. For any quantum state \(\rho\), let \(p_\rho(s):=\Tr(M_s\rho)\) denote the decoded distribution. If, for some \(\varepsilon\ge0\),
\begin{align}
\Tr(H_{\Dcal}[f]\rho)\ge f_{\max}-\varepsilon,\label{eq:near-optimal-decoded-expectation}
\end{align}
then
\[
\sum_{s\in\X_{\max}}p_\rho(s)
\ge
1-\frac{\varepsilon}{\Delta_\star}.
\]
The bound is informative when $\varepsilon<\Delta_\star$; otherwise its nontrivial lower bound is zero.
In particular, the same conclusion holds for a maximum-eigenvalue state \(\rho_{\max}\) whenever
\[
\lambda_{\max}(H_{\Dcal}[f])
\ge
f_{\max}-\varepsilon.
\]
\end{proposition}

\begin{proof}
    Decoder consistency gives
\[
\Tr(H_{\Dcal}[f]\rho)=\sum_s f(s)p_\rho(s).
\]
Hence,
\[
\begin{aligned}
f_{\max}-\Tr(H_{\Dcal}[f]\rho)
&=
\sum_s\bigl(f_{\max}-f(s)\bigr)p_\rho(s)\\
&\ge
\Delta_\star
\sum_{s\notin\X_{\max}}p_\rho(s).
\end{aligned}
\]By assumption, the left-hand side is at most \(\varepsilon\), and therefore
\[
\sum_{s\notin\X_{\max}}p_\rho(s)
\le
\frac{\varepsilon}{\Delta_\star}.
\]Taking the complement proves the claimed bound. For a maximum-eigenvalue state, use \(\Tr(H_{\Dcal}[f]\rho_{\max})=\lambda_{\max}(H_{\Dcal}[f])\).
\end{proof}

\begin{corollary}[Concentration on near-optimal solutions]
    For any \(\eta>0\), define
\[
\X_\eta:=\{s\in\X:f(s)\ge f_{\max}-\eta\}.
\]
If \(\rho\) satisfies Eq.~\eqref{eq:near-optimal-decoded-expectation}, then
\[
\sum_{s\in\X_\eta}p_\rho(s)
\ge
1-\frac{\varepsilon}{\eta}.
\]
\end{corollary}

\begin{proof}
If $s\notin\X_\eta$, then $f_{\max}-f(s)>\eta$. Repeating the previous loss calculation gives $f_{\max}-\Tr(H_{\Dcal}[f]\rho)\ge\eta\sum_{s\notin\X_\eta}p_\rho(s)$. The assumed bound then yields $\sum_{s\notin\X_\eta}p_\rho(s)\le\varepsilon/\eta$, and taking the complement proves the claim.
\end{proof}

\subsection{Test states and effective preservation rates}
\label{app:test-states}

We now make precise the meaning of a test state. Fix a decoder \(\Dcal\) and a set \(\mathcal I\) of nonconstant Fourier components of interest, and suppose that \(O_A^\Dcal:=\Dcal^\dagger(S_{C,A})\neq0\) for every \(A\in\mathcal I\). Write
\[
O_A^\Dcal=\kappa_A\widetilde O_A,
\qquad
\kappa_A=\|O_A^\Dcal\|_\infty,
\qquad
\|\widetilde O_A\|_\infty=1.
\]
\begin{definition}[Test-state family]
\label{def:test-state}
A family of density operators $\{\tau_s\}_{s\in\X}$ is called a test-state family for $\mathcal I$ if there exist signal rates $\mu_A\in[0,1]$ such that
\begin{align}
\Tr(\widetilde O_A\tau_s)
=
\mu_A s^A
\label{eq:test-state-definition}
\end{align}
for every $s\in\X$ and $A\in\mathcal I$.
\end{definition}

Equivalently, 
\begin{align}
  \Tr\!\left(\Dcal^\dagger(S_{C,A})\tau_s\right)
  =
  \alpha_A s^A,
  \qquad
  \alpha_A:=\kappa_A\mu_A.
  \label{eq:test-state-effective-rate}
\end{align}
Here, \(\kappa_A\) is the decoder-side preservation rate, \(\mu_A\) is the test-state signal rate, and \(\alpha_A\) is the resulting effective preservation rate.
Such a test-state family need not exist for every POVM decoder and every chosen component set. The lower bounds below are therefore conditional on the existence, or explicit construction, of physical states satisfying Eq.~\eqref{eq:test-state-definition} for the selected decoder block.
For the anticommuting-Pauli constructions used below, such test states can be constructed explicitly. Suppose that \(\{\Gamma_A^{(B)}\}_{A\in\Kcal_B}\) are traceless, mutually anticommuting Pauli operators and let \(K_B=|\Kcal_B|\). For each local configuration \(s_B\), define
\[
\tau_{s_B}^{(B)}
:=
\frac1{\dim\mathcal H_{Q_B}}
\left[
I+
\frac1{\sqrt{K_B}}
\sum_{A\in\Kcal_B}s_B^A\Gamma_A^{(B)}
\right].
\]
By mutual anticommutativity, the nonidentity term in brackets has operator norm one. Hence \(\tau_{s_B}^{(B)}\) is positive, and tracelessness of the Pauli operators gives \(\Tr\tau_{s_B}^{(B)}=1\).

Pauli orthogonality further gives
\begin{align}
  \Tr\!\left(\Gamma_A^{(B)}\tau_{s_B}^{(B)}\right)
  =
  \frac{1}{\sqrt{K_B}}s_B^A,
  \qquad A\in\Kcal_B,
  \label{eq:anticommuting-pauli-test-state-signal}
\end{align}
so the local signal rate is
\[
\mu_A^{(B)}=\frac1{\sqrt{K_B}}.
\]
Tensor products of the block states give a global test-state family, and the signal rates multiply across the blocks intersected by the component.

For the $(2,1)$ direct parity block, $K_B=3$ and $\tau_{s_1,s_2}=\frac12[I+(s_1X+s_2Y+s_1s_2Z)/\sqrt3]$, giving $\mu=1/\sqrt3$.
For the lower-bound argument, a full family indexed by all \(s\) is not necessary. It is sufficient to have a single physical state \(\tau_{s^*}\) satisfying the required component relations for a classical optimal solution \(s^*\). We therefore also refer to such a single witness as a test state. The test state need not be prepared by the algorithm; it serves only as a witness state for lower-bounding \(\lambda_{\max}(H_\Dcal[f])\) through the variational principle.
\paragraph{Specialization to term-complete decoders.}
\label{app:term-complete-guarantee}

For a term-complete decoder, every nonconstant Fourier component appearing in the objective is covered by the blockwise preserved sets. Its global preservation rate therefore factorizes as
\begin{align}
  \kappa_A
  =
  \prod_{B:A_B\neq\emptyset}\kappa_{A_B}^{(B)}.
  \label{eq:term-complete-kappa-product}
\end{align}
If a test state $\tau_{s^*}$ associated with a classical optimal solution satisfies 
\begin{align}
  \Tr\!\left[
    \left(
    \bigotimes_{B:A_B\neq\emptyset}\Gamma_{A_B}^{(B)}
    \right)
    \tau_{s^*}
  \right]
  =
  \mu_A(s^*)^A
  \label{eq:general-code-state-signal},
\end{align}
then $\alpha_A=(\prod_{B:A_B\neq\emptyset}\kappa_{A_B}^{(B)})\mu_A$, which can be substituted directly into Eq.~\eqref{eq:general-test-state-bound}.
An uncovered component of a non-term-complete decoder may be viewed as having effective preservation rate zero. If the decoder is term-complete and all nonconstant components share a common effective preservation rate \(\alpha\), Eq.~\eqref{eq:uniform-alpha-bound} yields
\begin{align}
  \lambda_{\max}(H_{\Dcal}[f])
  \ge
  C+\alpha(f_{\max}-C).
  \label{eq:term-complete-uniform-bound}
\end{align}
\subsection{MaxCut lower bound with distinct effective preservation rates for intra-block and inter-block edges}
\label{app:maxcut-two-rate-bound}

We derive a MaxCut lower bound for the case in which intra-block and inter-block edges have different effective preservation rates. This result is useful for decoders such as the symmetric direct parity block, whose preservation profile naturally splits into two rates.
For an unweighted graph $G=(V,E)$, let $\Cut(s)=|E|/2-(1/2)\sum_{(i,j)\in E}s_is_j$. Partition the edge set into intra-block edges $E_{\rm in}$ and inter-block edges $E_{\rm out}$, let $\lambda:=|E_{\rm out}|/|E|$, and write the classical optimum as $\OPT=(1/2+\varepsilon)|E|$ with $0\le\varepsilon\le1/2$.
Let $s^*$ be a classical optimal solution and define $u_{ij}:=-\frac12s_i^*s_j^*\in\{-1/2,1/2\}$. Since $\OPT=(1/2+\varepsilon)|E|$, we have $\sum_{(i,j)\in E}u_{ij}=\varepsilon|E|$.
Define $u_{\rm in}:=|E|^{-1}\sum_{(i,j)\in E_{\rm in}}u_{ij}$ and $u_{\rm out}:=|E|^{-1}\sum_{(i,j)\in E_{\rm out}}u_{ij}$. Then $u_{\rm in}+u_{\rm out}=\varepsilon$, with $-(1-\lambda)/2\le u_{\rm in}\le(1-\lambda)/2$ and $-\lambda/2\le u_{\rm out}\le\lambda/2$.
Let \(\tau_{s^*}\) be a test state associated with a classical optimal solution \(s^*\). Assume that the two-body components have effective preservation rates \(\alpha_{\rm in}\) and \(\alpha_{\rm out}\) on intra-block and inter-block edges, respectively, i.e.,
\begin{align}
  \Tr\!\left(\Dcal^\dagger(S_{C,\{i,j\}})\tau_{s^*}\right)
  =\begin{cases}
\alpha_{\rm in}s_i^*s_j^*, &(i,j)\in E_{\rm in},\\
\alpha_{\rm out}s_i^*s_j^*, &(i,j)\in E_{\rm out}.
\end{cases}
\end{align}
Evaluating the decoder-consistent Hamiltonian on \(\tau_{s^*}\) and applying the variational principle gives
\begin{align}
\frac{\lambda_{\max}(H_{\Dcal}[\Cut])}{|E|}
\ge
\frac12+\alpha_{\rm in}u_{\rm in}
+\alpha_{\rm out}u_{\rm out}.
\label{eq:maxcut-two-rate-linear-bound}
\end{align}
First assume $\alpha_{\rm in}\ge\alpha_{\rm out}$. Substituting $u_{\rm out}=\varepsilon-u_{\rm in}$ gives $1/2+\alpha_{\rm out}\varepsilon+(\alpha_{\rm in}-\alpha_{\rm out})u_{\rm in}$, which is minimized by the smallest allowed value of $u_{\rm in}$.
Therefore,
\[
u_{\rm in}\ge
\max\left\{-\frac{1-\lambda}{2},
\varepsilon-\frac{\lambda}{2}\right\},
\]and hence the stronger of the two resulting lower bounds is their maximum.
The case \(\alpha_{\rm in}<\alpha_{\rm out}\) follows symmetrically. Defining \(\Delta\alpha:=|\alpha_{\rm in}-\alpha_{\rm out}|\) and using decoder consistency for a maximum-eigenvalue state \(\rho_{\max}\), we obtain
\begin{widetext}
\begin{align}
\frac{\E_{s\sim\Dcal(\rho_{\max})}[\Cut(s)]}{\OPT}\ge
\max\left\{
\frac{1/2+\alpha_{\rm out}\varepsilon-\Delta\alpha(1-\lambda)/2}{1/2+\varepsilon},
\frac{1/2+\alpha_{\rm in}\varepsilon-\Delta\alpha\lambda/2}{1/2+\varepsilon}
\right\}\label{eq:maxcut-two-rate-ratio-bound}.
\end{align}
\end{widetext}

When there are no intra-block edges, the bound reduces to the common-rate result in the main text. Indeed, \(E_{\rm in}=\emptyset\) implies \(\lambda=1\) and \(u_{\rm in}=0\), so Eq.~\eqref{eq:maxcut-two-rate-linear-bound} becomes
\[
\frac{\lambda_{\max}(H_\Dcal[\Cut])}{|E|}
\ge
\frac12+\alpha_{\rm out}\varepsilon.
\]This is exactly Eq.~\eqref{eq:maxcut-alpha-bound-decoded} with the common effective preservation rate \(\alpha=\alpha_{\rm out}\), and gives the ratio \((1/2+\alpha_{\rm out}\varepsilon)/(1/2+\varepsilon)\).
For the symmetric $(2,1)$ direct parity block, $\alpha_{\rm in}=1/3$ and $\alpha_{\rm out}=1/9$. Substituting into Eq.~\eqref{eq:maxcut-two-rate-ratio-bound} gives
\begin{align}
  \frac{\E_{s\sim\Dcal(\rho_{\max})}[\Cut(s)]}{\OPT}
  &=
  \frac{\lambda_{\max}(H_{\Dcal}[\Cut])}{\OPT}
  \\
  &\ge
  \max\left\{
  \frac{7+2\lambda+2\varepsilon}{9+18\varepsilon},
  \frac{9-2\lambda+6\varepsilon}{9+18\varepsilon}
  \right\}.
  \label{eq:21-direct-parity-two-rate-bound}
\end{align}
Both terms in the maximum decrease with \(\varepsilon\in[0,1/2]\), so the worst case occurs at \(\varepsilon=1/2\). At this point, the two terms become \((4+\lambda)/9\) and \((6-\lambda)/9\), respectively. Since \(0\le\lambda\le1\), the latter is larger, yielding
\[
\frac{\E_{s\sim\Dcal(\rho_{\max})}[\Cut(s)]}{\OPT}
\ge
\frac{6-\lambda}{9}.
\]At \(\lambda=0\), all edges are intra-block and the bound is \(2/3\); at \(\lambda=1\), all edges are inter-block and the bound reduces to \(5/9\).
\section{Pullback calculations and decoder-block details}
\label{app:pullbacks-and-blocks}

\subsection{\texorpdfstring{Supplementary details on the \((n,n-1)\)-QRAC simultaneous-decoding POVM}{Supplementary details on the (n,n-1)-QRAC simultaneous-decoding POVM}}
\label{app:nnminus1-qrao}

The detailed $(n,n-1)$-QRAC construction follows Ref.~\cite{suzuki2026nnminus1qrac}. Here we summarize only the simultaneous-decoding POVM and pullback structure needed for the coefficient relations used in the main text.
Let $E_n:=\{x\in\{0,1\}^n:|x|\equiv0\pmod2\}$ and $O_n:=\{x\in\{0,1\}^n:|x|\equiv1\pmod2\}$. In the $(n,n-1)$-QRAC, each $x$ corresponds to a code state $\ket{\psi_x}$ on $n-1$ qubits, and both $\{\ket{\psi_x}:x\in E_n\}$ and $\{\ket{\psi_x}:x\in O_n\}$ form orthonormal bases.
We define the block simultaneous-decoding POVM by $M_x^{\mathrm{sim}}:=\frac12\ket{\psi_x}\bra{\psi_x}$. Since the even- and odd-parity code states each form a complete orthonormal basis, this is a valid POVM. Equivalently, this POVM combines into a single measurement the procedure of choosing the even- or odd-parity basis with probability \(1/2\) and then performing the corresponding projective measurement.
To determine the coefficient of a one-body component, let \(\{M_{0|k},M_{1|k}\}\) denote the binary POVM for reading the \(k\)-th bit, and define the corresponding observable \(O_k:=M_{0|k}-M_{1|k}\). Setting \(\nu:=(n-1)/n\), the construction of Ref.~\cite{suzuki2026nnminus1qrac} gives
\[
O_k=
\frac{1}{2\sqrt{\nu}}
\sum_{x\in\{0,1\}^n}
(-1)^{x_k}\ket{\psi_x}\bra{\psi_x}.
\]

The pullback of the one-body component is 
\begin{align}
  (\Dcal_B^{\mathrm{sim}})^\dagger(S_{C_B,\{k\}})
  &=
  \sum_{x\in\{0,1\}^n}(-1)^{x_k}M_x^{\mathrm{sim}}
  \nonumber\\
  &=
  \frac12
  \sum_{x\in\{0,1\}^n}(-1)^{x_k}\ket{\psi_x}\bra{\psi_x}
  \nonumber\\
  &=
  \sqrt\nu\,O_k .
  \label{eq:nnminus1-single-pullback-app}
\end{align}
Therefore its preservation rate is $\sqrt\nu=\sqrt{(n-1)/n}$.
When two variables \(i\) and \(j\) belong to different blocks, the pullback factorizes into the tensor product of their one-body pullbacks, giving
\begin{align}
  \Dcal^\dagger(S_{C,\{i,j\}})
  =
  \nu\,O_iO_j,
  \qquad B(i)\neq B(j).
  \label{eq:nnminus1-interblock-pair-pullback-app}
\end{align}
For two variables in the same block, the two-body pullback does not in general reduce to the simple product \(O_iO_j\). In the main text, we therefore restrict attention to interaction edges connecting different blocks, as ensured by the coloring constraint.
\subsection{Positivity budget for general block POVMs}
\label{app:positivity-and-pullback}

We now verify positivity of the general block POVM in Eq.~\eqref{eq:parity-preserving-block-povm}.
For each \(s_B\), define
\[
G_B(s_B):=\sum_{A\in\Kcal_B}\kappa_A^{(B)}s^A\Gamma_A^{(B)}.
\]Then \(M_{s_B}^{(B)}\ge0\) is equivalent to \(I+G_B(s_B)\ge0\). Since \(G_B(s_B)\) is Hermitian, the condition
\[
\|G_B(s_B)\|_\infty\le1
\]ensures that all its eigenvalues lie in \([-1,1]\), and hence \(I+G_B(s_B)\ge0\).

If the $\Gamma_A^{(B)}$ are mutually anticommuting Pauli operators, then 
\begin{align}
  G_B(s_B)^2
  &=
  \sum_A\left(\kappa_A^{(B)}\right)^2(s^A)^2\left(\Gamma_A^{(B)}\right)^2
  \\
  &+
  \sum_{A<A'}\kappa_A^{(B)}\kappa_{A'}^{(B)}s^As^{A'}
  \left(\Gamma_A^{(B)}\Gamma_{A'}^{(B)}+\Gamma_{A'}^{(B)}\Gamma_A^{(B)}\right)
  \nonumber\\
  &=
  \left[\sum_A\left(\kappa_A^{(B)}\right)^2\right]\Id,
\end{align}
because \((s^A)^2=1\), \((\Gamma_A^{(B)})^2=I\), and all cross terms vanish by anticommutativity.
Hence $\|G_B(s_B)\|_\infty=\sqrt{\sum_A(\kappa_A^{(B)})^2}$, and the sufficient positivity condition simplifies to
\[
\sum_{A\in\Kcal_B}
\left(\kappa_A^{(B)}\right)^2
\le1.
\]
\subsection{Randomized Pauli implementation}
\label{app:randomized-pauli-implementation}

The canonical POVM in Eq.~\eqref{eq:parity-preserving-block-povm} generally requires a direct implementation of a multi-outcome POVM. As an alternative, one may use a randomized implementation in which, at each shot, a single component \(A\in\Kcal_B\) is selected, the corresponding Pauli operator \(\Gamma_A^{(B)}\) is measured, and the block output \(s_B\) is generated by classical post-processing.
Suppose that
\begin{align}
  \sum_{A\in\Kcal_B}\kappa_A^{(B)}\le1.
  \label{eq:randomized-pauli-l1-budget}
\end{align}
Select component \(A\) with probability \(\kappa_A^{(B)}\) and measure \(\Gamma_A^{(B)}\), obtaining an outcome \(y\in\{\pm1\}\). Conditioned on \(y\), output \(s_B\) uniformly from the configurations satisfying \(s^A=y\). With the remaining probability \(1-\sum_A\kappa_A^{(B)}\), output \(s_B\) uniformly at random. The resulting measurement induces the POVM
\begin{align}
  M_{s_B}^{\mathrm{rand}}
  =
  \frac1{2^{|B|}}
  \left(
    \Id+
    \sum_{A\in\Kcal_B}
    \kappa_A^{(B)}s^A\Gamma_A^{(B)}
  \right)
  \label{eq:randomized-pauli-induced-povm}
\end{align}
and therefore reproduces exactly the same input-output statistics as the canonical POVM with the same coefficients.

Thus, for mutually anticommuting components, the canonical POVM is subject to the \(\ell_2\)-type positivity budget in Eq.~\eqref{eq:l2-positivity-budget}, whereas this randomized Pauli implementation requires the stronger \(\ell_1\)-type constraint in Eq.~\eqref{eq:randomized-pauli-l1-budget}. If \(K\) components are treated symmetrically, the canonical construction allows a common preservation rate as large as \(1/\sqrt K\), whereas the randomized implementation allows at most \(1/K\). The randomized scheme is therefore easier to implement using standard Pauli measurements and classical post-processing, at the cost of a smaller achievable preservation rate.
\subsection{\texorpdfstring{Derivation of the POVM induced by the magic rounding of Teramoto et al.}{Derivation of the POVM induced by the magic rounding of Teramoto et al.}}
\label{app:teramoto-derivation}

For each \(k=(k_1,k_2)\in\{\pm1\}^2\), define the tetrahedral direction \(\vec n_k=(k_1,k_2,k_1k_2)/\sqrt3\) and the binary projective measurement
\[
N_{\pm|k}=\frac12(I\pm\vec n_k\cdot\vec\sigma).
\]
We choose the magic basis indexed by \(k\) uniformly at random, with probability \(1/4\), and perform the corresponding binary projective measurement. If the outcome is \(+\), we output \(s=k\); if the outcome is \(-\), we choose uniformly from the three outputs other than \(k\). Thus, the classical post-processing probabilities are
\[
K(s|+,k)=\delta_{s,k},
\qquad
K(s|-,k)=\frac{1-\delta_{s,k}}{3}.
\]
Combining the random basis choice, projective measurement, and classical post-processing gives the effective POVM
\begin{align}
  M_s^{\mathrm T}
  &=
  \frac14\sum_k\sum_{r=\pm}K(s|r,k)N_{r|k}
  \nonumber\\
  &=
  \frac14\left(N_{+|s}+\frac13\sum_{k\neq s}N_{-|k}\right).
\end{align}
Using the tetrahedral symmetry \(\sum_k\vec n_k=0\),
\[
\sum_kN_{-|k}=2I,
\]and hence
\[
\sum_{k\neq s}N_{-|k}=2I-N_{-|s}.
\]Substituting this into the expression above gives
\[
\begin{aligned}
M_s^{\mathrm T}
&=
\frac14\left[
N_{+|s}+\frac13(2I-N_{-|s})
\right]\\
&=
\frac14\left(
I+\frac23\vec n_s\cdot\vec\sigma
\right)\\
&=
\frac14\left[
I+\frac{2}{3\sqrt3}
(s_1X+s_2Y+s_1s_2Z)
\right],
\end{aligned}
\]which reproduces Eq.~\eqref{eq:teramoto-induced-povm}.
\bibliography{../literature/refs}

\begin{thebibliography}{18}%
\makeatletter
\providecommand \@ifxundefined [1]{%
 \@ifx{#1\undefined}
}%
\providecommand \@ifnum [1]{%
 \ifnum #1\expandafter \@firstoftwo
 \else \expandafter \@secondoftwo
 \fi
}%
\providecommand \@ifx [1]{%
 \ifx #1\expandafter \@firstoftwo
 \else \expandafter \@secondoftwo
 \fi
}%
\providecommand \natexlab [1]{#1}%
\providecommand \enquote  [1]{``#1''}%
\providecommand \bibnamefont  [1]{#1}%
\providecommand \bibfnamefont [1]{#1}%
\providecommand \citenamefont [1]{#1}%
\providecommand \href@noop [0]{\@secondoftwo}%
\providecommand \href [0]{\begingroup \@sanitize@url \@href}%
\providecommand \@href[1]{\@@startlink{#1}\@@href}%
\providecommand \@@href[1]{\endgroup#1\@@endlink}%
\providecommand \@sanitize@url [0]{\catcode `\\12\catcode `\$12\catcode `\&12\catcode `\#12\catcode `\^12\catcode `\_12\catcode `\%12\relax}%
\providecommand \@@startlink[1]{}%
\providecommand \@@endlink[0]{}%
\providecommand \url  [0]{\begingroup\@sanitize@url \@url }%
\providecommand \@url [1]{\endgroup\@href {#1}{\urlprefix }}%
\providecommand \urlprefix  [0]{URL }%
\providecommand \Eprint [0]{\href }%
\providecommand \doibase [0]{https://doi.org/}%
\providecommand \selectlanguage [0]{\@gobble}%
\providecommand \bibinfo  [0]{\@secondoftwo}%
\providecommand \bibfield  [0]{\@secondoftwo}%
\providecommand \translation [1]{[#1]}%
\providecommand \BibitemOpen [0]{}%
\providecommand \bibitemStop [0]{}%
\providecommand \bibitemNoStop [0]{.\EOS\space}%
\providecommand \EOS [0]{\spacefactor3000\relax}%
\providecommand \BibitemShut  [1]{\csname bibitem#1\endcsname}%
\let\auto@bib@innerbib\@empty
%</preamble>
\bibitem [{\citenamefont {Farhi}\ \emph {et~al.}(2014)\citenamefont {Farhi}, \citenamefont {Goldstone},\ and\ \citenamefont {Gutmann}}]{farhi2014}%
  \BibitemOpen
  \bibfield  {author} {\bibinfo {author} {\bibfnamefont {E.}~\bibnamefont {Farhi}}, \bibinfo {author} {\bibfnamefont {J.}~\bibnamefont {Goldstone}},\ and\ \bibinfo {author} {\bibfnamefont {S.}~\bibnamefont {Gutmann}},\ }\bibfield  {title} {\bibinfo {title} {A quantum approximate optimization algorithm},\ }\href@noop {} {\bibfield  {journal} {\bibinfo  {journal} {arXiv preprint arXiv:1411.4028}\ } (\bibinfo {year} {2014})}\BibitemShut {NoStop}%
\bibitem [{\citenamefont {Hadfield}\ \emph {et~al.}(2019)\citenamefont {Hadfield}, \citenamefont {Wang}, \citenamefont {O’gorman}, \citenamefont {Rieffel}, \citenamefont {Venturelli},\ and\ \citenamefont {Biswas}}]{hadfield2019}%
  \BibitemOpen
  \bibfield  {author} {\bibinfo {author} {\bibfnamefont {S.}~\bibnamefont {Hadfield}}, \bibinfo {author} {\bibfnamefont {Z.}~\bibnamefont {Wang}}, \bibinfo {author} {\bibfnamefont {B.}~\bibnamefont {O’gorman}}, \bibinfo {author} {\bibfnamefont {E.~G.}\ \bibnamefont {Rieffel}}, \bibinfo {author} {\bibfnamefont {D.}~\bibnamefont {Venturelli}},\ and\ \bibinfo {author} {\bibfnamefont {R.}~\bibnamefont {Biswas}},\ }\bibfield  {title} {\bibinfo {title} {From the quantum approximate optimization algorithm to a quantum alternating operator ansatz},\ }\href@noop {} {\bibfield  {journal} {\bibinfo  {journal} {Algorithms}\ }\textbf {\bibinfo {volume} {12}},\ \bibinfo {pages} {34} (\bibinfo {year} {2019})}\BibitemShut {NoStop}%
\bibitem [{\citenamefont {Zhu}\ \emph {et~al.}(2022)\citenamefont {Zhu}, \citenamefont {Tang}, \citenamefont {Barron}, \citenamefont {Mayhall}, \citenamefont {Barnes},\ and\ \citenamefont {Economou}}]{zhu2022}%
  \BibitemOpen
  \bibfield  {author} {\bibinfo {author} {\bibfnamefont {L.}~\bibnamefont {Zhu}}, \bibinfo {author} {\bibfnamefont {H.~L.}\ \bibnamefont {Tang}}, \bibinfo {author} {\bibfnamefont {G.~S.}\ \bibnamefont {Barron}}, \bibinfo {author} {\bibfnamefont {N.~J.}\ \bibnamefont {Mayhall}}, \bibinfo {author} {\bibfnamefont {E.}~\bibnamefont {Barnes}},\ and\ \bibinfo {author} {\bibfnamefont {S.~E.}\ \bibnamefont {Economou}},\ }\bibfield  {title} {\bibinfo {title} {Adaptive quantum approximate optimization algorithm for solving combinatorial problems on a quantum computer},\ }\href@noop {} {\bibfield  {journal} {\bibinfo  {journal} {Physical Review Research}\ }\textbf {\bibinfo {volume} {4}},\ \bibinfo {pages} {033029} (\bibinfo {year} {2022})}\BibitemShut {NoStop}%
\bibitem [{\citenamefont {Ambainis}\ \emph {et~al.}(1999)\citenamefont {Ambainis}, \citenamefont {Nayak}, \citenamefont {Ta-Shma},\ and\ \citenamefont {Vazirani}}]{ambainis1999}%
  \BibitemOpen
  \bibfield  {author} {\bibinfo {author} {\bibfnamefont {A.}~\bibnamefont {Ambainis}}, \bibinfo {author} {\bibfnamefont {A.}~\bibnamefont {Nayak}}, \bibinfo {author} {\bibfnamefont {A.}~\bibnamefont {Ta-Shma}},\ and\ \bibinfo {author} {\bibfnamefont {U.}~\bibnamefont {Vazirani}},\ }\bibfield  {title} {\bibinfo {title} {Dense quantum coding and a lower bound for 1-way quantum automata},\ }in\ \href@noop {} {\emph {\bibinfo {booktitle} {Proceedings of the 31st Annual ACM Symposium on Theory of Computing (STOC)}}}\ (\bibinfo {year} {1999})\ pp.\ \bibinfo {pages} {376--383}\BibitemShut {NoStop}%
\bibitem [{\citenamefont {Nayak}(1999)}]{nayak1999}%
  \BibitemOpen
  \bibfield  {author} {\bibinfo {author} {\bibfnamefont {A.}~\bibnamefont {Nayak}},\ }\bibfield  {title} {\bibinfo {title} {Optimal lower bounds for quantum automata and random access codes},\ }in\ \href@noop {} {\emph {\bibinfo {booktitle} {40th Annual Symposium on Foundations of Computer Science (Cat. No. 99CB37039)}}}\ (\bibinfo {organization} {IEEE},\ \bibinfo {year} {1999})\ pp.\ \bibinfo {pages} {369--376}\BibitemShut {NoStop}%
\bibitem [{\citenamefont {Fuller}\ \emph {et~al.}(2024)\citenamefont {Fuller}, \citenamefont {Hadfield}, \citenamefont {Glick}, \citenamefont {Imamichi}, \citenamefont {Itoko}, \citenamefont {Thompson}, \citenamefont {Jiao}, \citenamefont {Kagele}, \citenamefont {Blom-Schieber}, \citenamefont {Raymond} \emph {et~al.}}]{fuller2024}%
  \BibitemOpen
  \bibfield  {author} {\bibinfo {author} {\bibfnamefont {B.}~\bibnamefont {Fuller}}, \bibinfo {author} {\bibfnamefont {C.}~\bibnamefont {Hadfield}}, \bibinfo {author} {\bibfnamefont {J.~R.}\ \bibnamefont {Glick}}, \bibinfo {author} {\bibfnamefont {T.}~\bibnamefont {Imamichi}}, \bibinfo {author} {\bibfnamefont {T.}~\bibnamefont {Itoko}}, \bibinfo {author} {\bibfnamefont {R.~J.}\ \bibnamefont {Thompson}}, \bibinfo {author} {\bibfnamefont {Y.}~\bibnamefont {Jiao}}, \bibinfo {author} {\bibfnamefont {M.~M.}\ \bibnamefont {Kagele}}, \bibinfo {author} {\bibfnamefont {A.~W.}\ \bibnamefont {Blom-Schieber}}, \bibinfo {author} {\bibfnamefont {R.}~\bibnamefont {Raymond}}, \emph {et~al.},\ }\bibfield  {title} {\bibinfo {title} {Approximate solutions of combinatorial problems via quantum relaxations},\ }\href@noop {} {\bibfield  {journal} {\bibinfo  {journal} {IEEE Transactions on Quantum Engineering}\ }\textbf {\bibinfo {volume} {5}},\ \bibinfo {pages} {1} (\bibinfo {year} {2024})}\BibitemShut {NoStop}%
\bibitem [{\citenamefont {Teramoto}\ \emph {et~al.}(2023)\citenamefont {Teramoto}, \citenamefont {Raymond}, \citenamefont {Wakakuwa},\ and\ \citenamefont {Imai}}]{teramoto2023}%
  \BibitemOpen
  \bibfield  {author} {\bibinfo {author} {\bibfnamefont {K.}~\bibnamefont {Teramoto}}, \bibinfo {author} {\bibfnamefont {R.}~\bibnamefont {Raymond}}, \bibinfo {author} {\bibfnamefont {E.}~\bibnamefont {Wakakuwa}},\ and\ \bibinfo {author} {\bibfnamefont {H.}~\bibnamefont {Imai}},\ }\bibfield  {title} {\bibinfo {title} {Quantum-relaxation based optimization algorithms: Theoretical extensions},\ }\href@noop {} {\bibfield  {journal} {\bibinfo  {journal} {arXiv preprint arXiv:2302.09481}\ } (\bibinfo {year} {2023})}\BibitemShut {NoStop}%
\bibitem [{\citenamefont {He}\ \emph {et~al.}(2025)\citenamefont {He}, \citenamefont {Raymond}, \citenamefont {Shaydulin},\ and\ \citenamefont {Pistoia}}]{he2025}%
  \BibitemOpen
  \bibfield  {author} {\bibinfo {author} {\bibfnamefont {Z.}~\bibnamefont {He}}, \bibinfo {author} {\bibfnamefont {R.}~\bibnamefont {Raymond}}, \bibinfo {author} {\bibfnamefont {R.}~\bibnamefont {Shaydulin}},\ and\ \bibinfo {author} {\bibfnamefont {M.}~\bibnamefont {Pistoia}},\ }\bibfield  {title} {\bibinfo {title} {Non-variational quantum random access optimization with alternating operator ansatz},\ }\href@noop {} {\bibfield  {journal} {\bibinfo  {journal} {Scientific Reports}\ }\textbf {\bibinfo {volume} {15}},\ \bibinfo {pages} {29191} (\bibinfo {year} {2025})}\BibitemShut {NoStop}%
\bibitem [{\citenamefont {Kondo}\ \emph {et~al.}(2025)\citenamefont {Kondo}, \citenamefont {Sato}, \citenamefont {Raymond},\ and\ \citenamefont {Yamamoto}}]{kondo2025}%
  \BibitemOpen
  \bibfield  {author} {\bibinfo {author} {\bibfnamefont {R.}~\bibnamefont {Kondo}}, \bibinfo {author} {\bibfnamefont {Y.}~\bibnamefont {Sato}}, \bibinfo {author} {\bibfnamefont {R.}~\bibnamefont {Raymond}},\ and\ \bibinfo {author} {\bibfnamefont {N.}~\bibnamefont {Yamamoto}},\ }\bibfield  {title} {\bibinfo {title} {Recursive quantum relaxation for combinatorial optimization problems},\ }\href@noop {} {\bibfield  {journal} {\bibinfo  {journal} {Quantum}\ }\textbf {\bibinfo {volume} {9}},\ \bibinfo {pages} {1594} (\bibinfo {year} {2025})}\BibitemShut {NoStop}%
\bibitem [{\citenamefont {Ambainis}\ \emph {et~al.}(2002)\citenamefont {Ambainis}, \citenamefont {Nayak}, \citenamefont {Ta-Shma},\ and\ \citenamefont {Vazirani}}]{ambainis2002}%
  \BibitemOpen
  \bibfield  {author} {\bibinfo {author} {\bibfnamefont {A.}~\bibnamefont {Ambainis}}, \bibinfo {author} {\bibfnamefont {A.}~\bibnamefont {Nayak}}, \bibinfo {author} {\bibfnamefont {A.}~\bibnamefont {Ta-Shma}},\ and\ \bibinfo {author} {\bibfnamefont {U.}~\bibnamefont {Vazirani}},\ }\bibfield  {title} {\bibinfo {title} {Dense quantum coding and quantum finite automata},\ }\href@noop {} {\bibfield  {journal} {\bibinfo  {journal} {Journal of the ACM (JACM)}\ }\textbf {\bibinfo {volume} {49}},\ \bibinfo {pages} {496} (\bibinfo {year} {2002})}\BibitemShut {NoStop}%
\bibitem [{\citenamefont {Imamichi}\ and\ \citenamefont {Raymond}(2018)}]{imamichi2018}%
  \BibitemOpen
  \bibfield  {author} {\bibinfo {author} {\bibfnamefont {T.}~\bibnamefont {Imamichi}}\ and\ \bibinfo {author} {\bibfnamefont {R.}~\bibnamefont {Raymond}},\ }\href@noop {} {\bibinfo {title} {Constructions of quantum random access codes}},\ \bibinfo {howpublished} {Asian Quantum Information Symposium (AQIS)} (\bibinfo {year} {2018}),\ \bibinfo {note} {available at \url{http://www.ngc.is.ritsumei.ac.jp/~ger/static/AQIS18/OnlineBooklet/122}}\BibitemShut {NoStop}%
\bibitem [{\citenamefont {Suzuki}(2026)}]{suzuki2026nnminus1qrac}%
  \BibitemOpen
  \bibfield  {author} {\bibinfo {author} {\bibfnamefont {T.}~\bibnamefont {Suzuki}},\ }\bibfield  {title} {\bibinfo {title} {Analytical construction of {$(n,n-1)$} quantum random access codes saturating the conjectured bound},\ }\href@noop {} {\bibfield  {journal} {\bibinfo  {journal} {arXiv preprint arXiv:2601.19190}\ } (\bibinfo {year} {2026})}\BibitemShut {NoStop}%
\bibitem [{\citenamefont {Kondo}\ \emph {et~al.}(2026)\citenamefont {Kondo}, \citenamefont {Sato}, \citenamefont {Yano}, \citenamefont {Maeda}, \citenamefont {Ito},\ and\ \citenamefont {Yamamoto}}]{kondo2026rac}%
  \BibitemOpen
  \bibfield  {author} {\bibinfo {author} {\bibfnamefont {R.}~\bibnamefont {Kondo}}, \bibinfo {author} {\bibfnamefont {Y.}~\bibnamefont {Sato}}, \bibinfo {author} {\bibfnamefont {H.}~\bibnamefont {Yano}}, \bibinfo {author} {\bibfnamefont {Y.}~\bibnamefont {Maeda}}, \bibinfo {author} {\bibfnamefont {K.}~\bibnamefont {Ito}},\ and\ \bibinfo {author} {\bibfnamefont {N.}~\bibnamefont {Yamamoto}},\ }\bibfield  {title} {\bibinfo {title} {Random access codes: Explicit constructions, optimality, and classical-quantum gaps},\ }\href@noop {} {\bibfield  {journal} {\bibinfo  {journal} {arXiv preprint arXiv:2604.21274}\ } (\bibinfo {year} {2026})}\BibitemShut {NoStop}%
\bibitem [{\citenamefont {Matsuyama}\ \emph {et~al.}(2024)\citenamefont {Matsuyama}, \citenamefont {Huang}, \citenamefont {Nishimura},\ and\ \citenamefont {Yamashiro}}]{matsuyama2024internal}%
  \BibitemOpen
  \bibfield  {author} {\bibinfo {author} {\bibfnamefont {H.}~\bibnamefont {Matsuyama}}, \bibinfo {author} {\bibfnamefont {W.-h.}\ \bibnamefont {Huang}}, \bibinfo {author} {\bibfnamefont {K.}~\bibnamefont {Nishimura}},\ and\ \bibinfo {author} {\bibfnamefont {Y.}~\bibnamefont {Yamashiro}},\ }\bibfield  {title} {\bibinfo {title} {Efficient internal strategies in quantum relaxation based branch-and-bound},\ }\href@noop {} {\bibfield  {journal} {\bibinfo  {journal} {arXiv preprint arXiv:2405.00935}\ } (\bibinfo {year} {2024})}\BibitemShut {NoStop}%
\bibitem [{\citenamefont {{Qiskit Optimization contributors}}(2023)}]{qiskitOptimizationQRAOEncoding2023}%
  \BibitemOpen
  \bibfield  {author} {\bibinfo {author} {\bibnamefont {{Qiskit Optimization contributors}}},\ }\href@noop {} {\bibinfo {title} {{Qiskit Optimization}: \texttt{quantum\_random\_access\_encoding.py}}},\ \bibinfo {howpublished} {GitHub repository, commit 09134bc} (\bibinfo {year} {2023}),\ \bibinfo {note} {\url{https://github.com/qiskit-community/qiskit-optimization/blob/09134bc0c1da7dd4852613c24cfeb5c32dcbcb54/qiskit_optimization/algorithms/qrao/quantum_random_access_encoding.py}}\BibitemShut {NoStop}%
\bibitem [{\citenamefont {Sharma}\ \emph {et~al.}(2024)\citenamefont {Sharma}, \citenamefont {Yan}, \citenamefont {Lau},\ and\ \citenamefont {Raymond}}]{sharma2024knapsack}%
  \BibitemOpen
  \bibfield  {author} {\bibinfo {author} {\bibfnamefont {M.}~\bibnamefont {Sharma}}, \bibinfo {author} {\bibfnamefont {J.}~\bibnamefont {Yan}}, \bibinfo {author} {\bibfnamefont {H.~C.}\ \bibnamefont {Lau}},\ and\ \bibinfo {author} {\bibfnamefont {R.}~\bibnamefont {Raymond}},\ }\bibfield  {title} {\bibinfo {title} {Quantum relaxation for solving multiple knapsack problems},\ }in\ \href@noop {} {\emph {\bibinfo {booktitle} {2024 IEEE International Conference on Quantum Computing and Engineering (QCE)}}}\ (\bibinfo {organization} {IEEE},\ \bibinfo {year} {2024})\ pp.\ \bibinfo {pages} {692--698}\BibitemShut {NoStop}%
\bibitem [{\citenamefont {Sharma}\ and\ \citenamefont {Lau}(2025)}]{sharma2025comparative}%
  \BibitemOpen
  \bibfield  {author} {\bibinfo {author} {\bibfnamefont {M.}~\bibnamefont {Sharma}}\ and\ \bibinfo {author} {\bibfnamefont {H.~C.}\ \bibnamefont {Lau}},\ }\bibfield  {title} {\bibinfo {title} {A comparative study of quantum optimization techniques for solving combinatorial optimization benchmark problems},\ }\href@noop {} {\bibfield  {journal} {\bibinfo  {journal} {arXiv preprint arXiv:2503.12121}\ } (\bibinfo {year} {2025})}\BibitemShut {NoStop}%
\bibitem [{\citenamefont {Renes}\ \emph {et~al.}(2004)\citenamefont {Renes}, \citenamefont {Blume-Kohout}, \citenamefont {Scott},\ and\ \citenamefont {Caves}}]{renes2004sic}%
  \BibitemOpen
  \bibfield  {author} {\bibinfo {author} {\bibfnamefont {J.~M.}\ \bibnamefont {Renes}}, \bibinfo {author} {\bibfnamefont {R.}~\bibnamefont {Blume-Kohout}}, \bibinfo {author} {\bibfnamefont {A.~J.}\ \bibnamefont {Scott}},\ and\ \bibinfo {author} {\bibfnamefont {C.~M.}\ \bibnamefont {Caves}},\ }\bibfield  {title} {\bibinfo {title} {Symmetric informationally complete quantum measurements},\ }\href@noop {} {\bibfield  {journal} {\bibinfo  {journal} {Journal of Mathematical Physics}\ }\textbf {\bibinfo {volume} {45}},\ \bibinfo {pages} {2171} (\bibinfo {year} {2004})}\BibitemShut {NoStop}%
\end{thebibliography}%

\end{document}